\documentclass[11pt,finalcls,onecolumn]{IEEEtran}
\usepackage{graphicx}
\usepackage{amsmath}
\usepackage{mathrsfs}
\usepackage{mathtools}
\usepackage{amssymb}
\usepackage{float}
\usepackage{url}
\usepackage{verbatim}
\usepackage{enumerate}
\usepackage{layout}

\usepackage[bottom=0.79in]{geometry}

\newtheorem{theorem}{Theorem}
\newtheorem{lemma}{Lemma}
\newtheorem{proposition}{Proposition}
\newtheorem{corollary}{Corollary}

\newtheorem{definition}{Definition}
\newtheorem{example}{Example}
\newtheorem{remark}{Remark}
\newcommand{\prob}{\ensuremath{\mathbb{P}}}
\newcommand{\naturals}{\ensuremath{\mathbb{N}}}
\newcommand{\Reals}{\ensuremath{\mathbb{R}}}
\newcommand{\expectation}{\ensuremath{\mathbb{E}}}
\newcommand{\Var}{\mathrm{Var}}

\newcommand{\set}{\ensuremath{\mathcal}}

\DeclareMathOperator*{\esssup}{ess\,sup}

\begin{document}
\thispagestyle{empty}
\setcounter{page}{1}
\setlength{\baselineskip}{1.15\baselineskip}

\title{\huge{On Relations Between the Relative Entropy
and $\chi^2$-Divergence, Generalizations and Applications}\\[0.2cm]}
\author{Tomohiro Nishiyama \qquad Igal Sason
\thanks{T. Nishiyama is an independent researcher, Tokyo, Japan
(e-mail: htam0ybboh@gmail.com).\newline
I. Sason is with the Andrew \& Erna Faculty of Electrical Engineering,
Technion - Israel Institute of Technology, Haifa 3200003, Israel
(e-mail: sason@ee.technion.ac.il).\newline
For citation: T. Nishiyama and I. Sason, ``On relations between the relative entropy and $\chi^2$-divergence,
generalizations and applications,'' {\em Entropy}, vol.~22, no.~5, paper~563, pp.~1--36, May 2020.}}

\maketitle
\thispagestyle{empty}

\begin{abstract}
The relative entropy and chi-squared divergence are fundamental
divergence measures in information theory and statistics. This paper is focused on
a study of integral relations between the two divergences, the
implications of these relations, their information-theoretic applications,
and some generalizations pertaining to the rich class of
$f$-divergences. Applications that are studied in this paper
refer to lossless compression, the method of types and large deviations,
strong~data-processing inequalities, bounds on contraction coefficients
and maximal correlation, and~the convergence rate to stationarity
of a type of discrete-time Markov chains.
\end{abstract}


\section{Introduction}
\label{section: introduction}

The relative entropy (also known as the Kullback--Leibler divergence \cite{SKRAL51}) and the
chi-squared divergence \cite{Pearson1900x} are divergence measures which play a key
role in information theory, statistics, learning, signal processing and other theoretical
and applied branches of mathematics. These divergence measures are fundamental in problems
pertaining to source and channel coding, combinatorics and large deviations theory,
goodness-of-fit and independence tests in statistics, expectation-maximization
iterative algorithms for estimating a distribution from an incomplete data, and other
sorts of problems (the reader is referred to the tutorial paper by Csisz\'{a}r and
Shields \cite{CsiszarS_FnT}).
They both belong to an important class of divergence measures, defined by means of
convex functions $f$, and named $f$-divergences (\cite{AliS},
\cite{Csiszar63}--\cite{Csiszar72}). In addition to the relative entropy and the chi-squared
divergence, this class unifies other useful divergence measures such as the total
variation distance in functional analysis, and it is also closely related to the
R\'{e}nyi divergence which generalizes the relative entropy (\cite{ErvenH14},
\cite{Renyientropy}). In general, $f$-divergences (defined in Section~\ref{section: preliminaries})
are attractive since they satisfy pleasing features such as the data-processing
inequality, convexity, (semi)continuity and duality properties, and they therefore
find nice applications in information theory and statistics (see, e.g.,
\cite{Csiszar67a, Csiszar72, LieseV_book87, LieseV_IT2006}).

In this work, we study integral relations between the relative entropy
and the chi-squared divergence, implications of these relations, and some of
their information-theoretic applications. Some generalizations
which apply to the class of $f$-divergences are also explored in details.
In this context, it should be noted that integral representations of
general $f$-divergences, expressed as a function of either the DeGroot
statistical information \cite{DeGroot62}, the $E_\gamma$-divergence (a
parametric sub-class of $f$-divergences which generalizes the total variation distance,
\cite[p.~2314]{PPV10}) and the relative information spectrum, have been derived
in \cite[Section~5]{LieseV_IT2006}, \cite[Section~7.B]{ISSV16} and
\cite[Section~3]{Sason18} respectively.

Applications in this paper are related to lossless source compression,
large deviations by the method of types, and strong data-processing inequalities.
Relevant background for each of these applications is provided to make the
presentation self contained.

We next outline the paper contributions, and the structure
of our manuscript.

\subsection{Paper Contributions}
This work starts by introducing integral relations between the relative entropy and
the chi-squared divergence, and some inequalities which relate these two divergences
(see Theorem~\ref{thm: KL and chi^2}, its corollaries,
and Proposition~\ref{prop: f-div. ineq.}). It continues with a study of the
implications and generalizations of these relations, pertaining to the rich class
of $f$-divergences. One implication leads to a tight lower bound on the relative
entropy between a pair of probability measures, expressed as a function of the means and
variances under these measures (see Theorem~\ref{theorem: KL LB}). A second implication
of Theorem~\ref{thm: KL and chi^2} leads to an upper bound on a skew divergence (see
Theorem~\ref{thm: F} and Corollary~\ref{corollary: UB on relative entropy}). Due to the
concavity of the Shannon entropy, let the concavity deficit of the entropy function be defined
as the non-negative difference between the entropy of a convex combination of distributions
and the convex combination of the entropies of these distributions.
Then, Corollary~\ref{corollary: entropy of a mixture of distributions} provides an upper
bound on this deficit, expressed as a function of the pairwise relative entropies
between all pairs of distributions. Theorem~\ref{theorem: monotonic_f_divergence_seq}
provides a generalization of Theorem~\ref{thm: KL and chi^2} to the class of $f$-divergences.
It recursively constructs non-increasing sequences of $f$-divergences and,
as a consequence of Theorem~\ref{theorem: monotonic_f_divergence_seq} followed by the
usage of polylogairthms, Corollary~\ref{corollary: polylogarithm} provides a
generalization of the useful integral relation in Theorem~\ref{thm: KL and chi^2} between
the relative entropy and the chi-squared divergence.
Theorem~\ref{thm: gen. Csiszar} relates probabilities of sets to $f$-divergences,
generalizing a known and useful result by Csisz\'{a}r for the relative entropy. In respect to
Theorem~\ref{thm: KL and chi^2}, the integral relation between
the relative entropy and the chi-squared divergence has been independently derived in \cite{Allerton19},
which also derived an alternative upper bound on the concavity deficit of the entropy as a function
of total variational distances (differing from the bound in Corollary~\ref{corollary: entropy of a mixture of distributions},
which depends on pairwise relative entropies). The interested reader is referred to \cite{Allerton19},
with a preprint of the extended version in~\cite{arXiv20}, and to \cite{Audenaert} where the connections
in Theorem~\ref{thm: KL and chi^2} were originally discovered in the quantum setting.

The second part of this work studies information-theoretic applications of the above results.
These are ordered by starting from the relatively simple applications, and ending at the more complicated ones.
The first one includes a bound on the redundancy of the Shannon code for universal lossless compression with
discrete memoryless sources, used in conjunction with Theorem~\ref{thm: F} (see
Section~\ref{subsec: UB on RE - Poisson}). An application of Theorem~\ref{theorem: KL LB} in the
context of the method of types and large deviations analysis is then studied in Section~\ref{subsec: method of types},
providing non-asymptotic bounds which lead to a closed-form expression as a function of the Lambert $W$ function
(see Proposition~\ref{prop: W function}).
Strong data-processing inequalities with bounds on contraction coefficients of skew divergences are provided
in Theorem~\ref{thm: contraction coefficients}, Corollary~\ref{corollary: bounds - contraction coefficient}
and Proposition~\ref{proposition: bounds - contraction coefficient}. Consequently, non-asymptotic bounds
on the convergence to stationarity of time-homogeneous, irreducible and reversible discrete-time Markov
chains with finite state spaces are obtained by relying on our bounds on the contraction coefficients of skew
divergences (see Theorem~\ref{thm: MC}). The exact asymptotic convergence rate is also obtained in
Corollary~\ref{corollary: lim. of n-th root}. Finally, a property of maximal correlations is obtained in
Proposition~\ref{prop: max. correlation} as an application of our starting point on the integral relation
between the relative entropy and the chi-squared divergence.

\subsection{Paper Organization}
This paper is structured as follows. Section~\ref{section: preliminaries}
presents notation and preliminary material which is necessary
for, or otherwise related to, the exposition of this work.
Section~\ref{section: main results} refers to the developed relations between divergences,
and Section~\ref{section: applications} studies information-theoretic applications.
Proofs of the results in Sections~\ref{section: main results} and~\ref{section: applications}
(except for short proofs) are deferred to Section~\ref{section: proofs}.

\section{Preliminaries and Notation}
\label{section: preliminaries}
This subsection provides definitions of divergence measures which are
used in this paper, and it also provides relevant notation.
\begin{definition} \label{def:fD} \cite[p.~4398]{LieseV_IT2006}
Let $P$ and $Q$ be probability measures, let $\mu$ be a
dominating measure of $P$ and $Q$ (i.e., $P, Q \ll \mu$),
and let $p := \frac{\mathrm{d}P}{\mathrm{d}\mu}$ and
$q := \frac{\mathrm{d}Q}{\mathrm{d}\mu}$ be the densities of
$P$ and $Q$ with respect to $\mu$. The {\em $f$-divergence}
from $P$ to $Q$ is given by
\begin{align} \label{eq:fD}
D_f(P\|Q) := \int q \, f \Bigl(\frac{p}{q}\Bigr) \, \mathrm{d}\mu,
\end{align}
where
\begin{align}
& f(0) := \underset{t \to 0^+}{\lim} \, f(t), \quad
0 f\biggl(\frac{0}{0}\biggr) := 0, \\[0.1cm]
& 0 f\biggl(\frac{a}{0}\biggr)
:= \lim_{t \to 0^+} \, t f\biggl(\frac{a}{t}\biggr)
= a \lim_{u \to \infty} \frac{f(u)}{u}, \quad a>0.
\end{align}
It should be noted that the right side of \eqref{eq:fD}
does not depend on the dominating measure $\mu$.
\end{definition}

\vspace*{0.1cm}
Throughout the paper, we denote by $1\{{\mathrm{relation}}\}$ the indicator function;
it is equal to~1 if the relation is true, and it is equal to~0 otherwise.
Throughout the paper, unless indicated explicitly, logarithms have an arbitrary
common base (that is larger than~1), and $\exp(\cdot)$ indicates the inverse function
of the logarithm with that base.

\begin{definition} \label{def:relative-entropy}  \cite{SKRAL51}
The {\em relative entropy} is the
$f$-divergence with $f(t) := t \log t$
for $t >0$,
\begin{align}
\label{eq: 1st KL divergence}
D(P\|Q) &:= D_f(P\|Q) \\
\label{eq2: 1st KL divergence}
&= \int p \, \log \frac{p}{q} \; \mathrm{d}\mu.
\end{align}
\end{definition}

\begin{definition} \label{def:TV}
The {\em total variation distance} between probability
measures $P$ and $Q$ is the $f$-divergence from $P$ to $Q$ with
$f(t) := |t-1|$ for all $t \geq 0$. It is a symmetric $f$-divergence,
denoted by $|P-Q|$, which is given by
\begin{align}
\label{eq1: TV distance}
|P-Q| &:=  D_f(P\|Q) \\
\label{eq2: TV distance}
&= \int |p-q| \, \mathrm{d}\mu.
\end{align}
\end{definition}

\begin{definition} \label{def:chi-squared}  \cite{Pearson1900x}
The {\em chi-squared divergence} from $P$ to $Q$
is defined to be the $f$-divergence in \eqref{eq:fD} with
$f(t) := (t-1)^2$ or $f(t) := t^2-1$ for all $t>0$,
\begin{align}
\label{eq: chi-square 1}
\chi^2(P\|Q) &:= D_f(P\|Q) \\
\label{eq: chi-square 1b}
&= \int \frac{(p-q)^2}{q} \, \mathrm{d}\mu
= \int \frac{p^2}{q} \, \mathrm{d}\mu - 1.
\end{align}
\end{definition}

\vspace*{0.2cm}
The R\'{e}nyi divergence, a generalization of the relative entropy, was introduced by
R\'{e}nyi \cite{Renyientropy} in the special case of finite alphabets. Its general
definition is given as follows (see, e.g., \cite{ErvenH14}).
\begin{definition} \label{def:RD}  \cite{Renyientropy}
Let $P$ and $Q$ be probability measures on $\set{X}$ dominated by $\mu$,
and let their densities be respectively denoted by $p = \frac{\mathrm{d}P}{\mathrm{d}\mu}$
and $q = \frac{\mathrm{d}Q}{\mathrm{d}\mu}$.
The {\em R\'{e}nyi divergence} of order $\alpha \in [0, \infty]$ is defined as follows:
\begin{itemize}
\item
If $\alpha \in (0,1) \cup (1, \infty) $, then
\begin{align} \label{eq:RD0}
D_{\alpha}(P\|Q) &=
\frac1{\alpha-1} \; \log \mathbb{E} \left[
p^\alpha(Z) \, q^{1-\alpha}(Z) \right] \\[0.1cm]
&= \frac1{\alpha-1} \; \log \,
\sum_{x \in \set{X}} P^\alpha (x) \,
Q^{1-\alpha} (x), \label{eq:RD1}
\end{align}
where $Z \sim \mu$ in \eqref{eq:RD0}, and \eqref{eq:RD1} holds if
$\set{X}$ is a discrete set.
\item By the continuous extension of $D_{\alpha}(P \| Q)$,
\begin{align}
\label{eq: d0}
& D_0(P \| Q) = \underset{\set{A}: P(\set{A})=1}{\max} \log \frac1{Q(\set{A})}, \\
\label{def:d1}
& D_1(P\|Q) = D(P\|Q), \\
\label{def:dinf}
& D_{\infty}(P\|Q) = \log \, \esssup \frac{p(Z)}{q(Z)}.
\end{align}
\end{itemize}
\end{definition}

The second-order R\'{e}nyi divergence and the chi-squared divergence are related as follows:
\begin{align} \label{RD2}
D_2(P\|Q) = \log\bigl(1+\chi^2(P\|Q)\bigr),
\end{align}
and the relative entropy and chi-squared divergence satisfy (see, e.g.,
\cite[Theorem~5]{GibbsSu02})
\begin{align} \label{RD2-chi2}
D(P\|Q) \leq \log \bigl(1 + \chi^2(P\|Q)\bigr).
\end{align}
Inequality~\eqref{RD2-chi2} readily follows from \eqref{def:d1}, \eqref{RD2},
and since $D_{\alpha}(P\|Q)$ is monotonically increasing in $\alpha \in (0, \infty)$
(see \cite[Theorem~3]{ErvenH14}). A tightened version of \eqref{RD2-chi2}, introducing
an improved and locally-tight upper bound on $D(P\|Q)$ as a function of $\chi^2(P\|Q)$
and $\chi^2(Q\|P)$, is introduced in \cite[Theorem~20]{ISSV16}. Another sharpened version
of \eqref{RD2-chi2} is derived in \cite[Theorem~11]{ISSV16} under the assumption of
a bounded relative information. Furthermore, under the latter assumption, tight upper
and lower bounds on the ratio $\frac{D(P\|Q)}{\chi^2(P\|Q)}$ are obtained in \cite[(169)]{ISSV16}.

\begin{definition}\cite{GV01}  \label{definition: GV divergence}
The {\em Gy\"{o}rfi--Vajda divergence} of order $s \in [0,1]$ is an $f$-divergence with
\begin{align} \label{eq: f - GV01}
f(t) = \phi_s(t) := \frac{(t-1)^2}{s + (1-s)t}, \quad t \geq 0.
\end{align}
Vincze--Le Cam distance (also known as the triangular discrimination) (\cite{Le Cam86,Vincze81})
is a special case with $s = \tfrac12$.
\end{definition}

In view of \eqref{eq:fD}, \eqref{eq: chi-square 1b} and \eqref{eq: f - GV01},
it can be verified that the Gy\"{o}rfi--Vajda divergence is related to the chi-squared
divergence as follows:
\begin{align} \label{eq: GV01-chi^2}
D_{\phi_s}(P\|Q) =
\begin{dcases}
\frac1{s^2} \cdot \chi^2\bigl(P \, \| \, (1-s)P + sQ \bigr),
& \quad s \in (0,1], \\
\chi^2(Q\|P), & \quad s = 0.
\end{dcases}
\end{align}
Hence,
\begin{align}
\label{f-div03}
& D_{\phi_1}(P\|Q) = \chi^2(P\|Q), \\
\label{f-div04}
& D_{\phi_0}(P\|Q) = \chi^2(Q\|P).
\end{align}

\section{Relations Between Divergences}
\label{section: main results}

We introduce in this section new results on the relations between
the relative entropy and the chi-squared divergence, and some of their implications
and generalizations. Information-theoretic applications of these results
are studied in the next section.

\subsection{Relations Between the Relative Entropy and the Chi-Squared Divergence}
\label{subsec: relative entropy and chi-squared divergence}

The following result relates the relative entropy and the chi-squared divergence,
which are two fundamental divergence measures in information theory and
statistics. This result was recently obtained in an equivalent form in
\cite[(12)]{Allerton19} (it is noted that this identity was also independently
derived by the coauthors in two separate un-published works in \cite[(16)]{Nishiyama19}
and \cite{Sason_June18}). It should be noted that these connections between
divergences in the quantum setting were originally discovered in \cite[Theorem~6]{Audenaert}.
Beyond serving as an interesting relation between these two fundamental
divergence measures, it is introduced here for the following reasons:
\begin{enumerate}[1)]
\item New consequences and applications of it are obtained, including
new shorter proofs of some known results;
\item An interesting extension in Section~\ref{subsec: monotonic_sequence}
provides new relations between $f$-divergences.
\end{enumerate}

\begin{theorem} \label{thm: KL and chi^2}
Let $P$ and $Q$ be probability measures defined on a measurable space
$(\set{X}, \mathscr{F})$, and let
\begin{align} \label{mixture PMF}
R_\lambda := (1-\lambda) P + \lambda Q, \quad \lambda \in [0,1]
\end{align}
be the convex combination of $P$ and $Q$. Then, for all $\lambda \in [0,1]$,
\begin{align} \label{identity: KL}
\tfrac1{\log \mathrm{e}} \, D(P\| R_\lambda)
&= \int_0^\lambda \chi^2(P \| R_s) \; \frac{\mathrm{d}s}{s}, \\[0.1cm]
\label{identity: chi2}
\tfrac12 \, \lambda^2 \, \chi^2(R_{1-\lambda} \| Q)
&= \int_0^\lambda \chi^2(R_{1-s} \| Q) \; \frac{\mathrm{d}s}{s}.
\end{align}
\end{theorem}
\begin{proof}
See Section~\ref{subsection: proof KL-chi2}.
\end{proof}

A specialization of Theorem~\ref{thm: KL and chi^2} by letting $\lambda = 1$
gives the following identities.
\begin{corollary} \label{corollary: RE and chi2}
\begin{align}
\label{identity2: KL}
\tfrac1{\log \mathrm{e}} \, D(P\| Q)
= \int_0^1 \chi^2(P \, \| \, (1-s)P + sQ) \; \frac{\mathrm{d}s}{s}, \\[0.1cm]
\label{identity2: chi2}
\tfrac12 \, \chi^2(P \| Q)
= \int_0^1 \chi^2(sP + (1-s)Q \, \| \, Q) \; \frac{\mathrm{d}s}{s}.
\end{align}
\end{corollary}

\begin{remark}
The substitution $s := \frac1{1+t}$ transforms \eqref{identity2: KL} to
\cite[Eq. (31)]{MakurP18}, i.e.,
\begin{align} \label{eq:MakurP18}
\tfrac1{\log \mathrm{e}} \, D(P\| Q)
= \int_0^\infty \chi^2\biggl(P \, \| \, \frac{tP+Q}{1+t} \biggr) \, \frac{\mathrm{d}t}{1+t}.
\end{align}
\end{remark}

In view of \eqref{eq: GV01-chi^2} and \eqref{mixture PMF}, an equivalent form of
\eqref{identity: KL} and \eqref{identity2: KL} is given as follows.
\begin{corollary}
\begin{align}
\label{identity3: KL}
\tfrac1{\log \mathrm{e}} \, D(P\| R_\lambda)
&= \int_0^\lambda  s D_{\phi_s}(P \| Q) \, \mathrm{d}s, \quad \lambda \in [0,1], \\[0.1cm]
\label{identity4: KL}
\tfrac1{\log \mathrm{e}} \, D(P\| Q)
&= \int_0^1  s D_{\phi_s}(P \| Q) \, \mathrm{d}s,
\end{align}
where $\phi_s$ in the right sides of \eqref{identity3: KL} and \eqref{identity4: KL} is
given in \eqref{eq: f - GV01}.
\end{corollary}

By Corollary~\ref{corollary: RE and chi2}, we obtain original and simple proofs
of new and old $f$-divergence inequalities.
\begin{proposition} \label{prop: f-div. ineq.}
({\em $f$-divergence inequalities}).
\begin{enumerate}[1)]
\item Pinsker's inequality:
\begin{align}  \label{Pinsker}
D(P \| Q) \geq \tfrac12  |P-Q|^2 \log \mathrm{e}.
\end{align}

\item
\begin{align}  \label{RD2-chi2 new}
\tfrac1{\log \mathrm{e}} \, D(P\|Q) \leq \tfrac13 \, \chi^2(P\|Q) + \tfrac16 \, \chi^2(Q\|P).
\end{align}
Furthermore, let $\{P_n\}$ be a sequence of probability measures that is defined on
a measurable space $(\set{X}, \mathscr{F})$, and which converges to a probability
measure $P$ in the sense that
\begin{align} \label{lim 1}
\lim_{n \to \infty} \esssup \frac{\mathrm{d}P_n}{\mathrm{d}P} \, (X) = 1,
\end{align}
with $X \sim P$. Then, \eqref{RD2-chi2 new} is locally tight in the sense that its
both sides converge to~0, and
\begin{align} \label{locally tight}
\lim_{n \to \infty} \frac{\tfrac13 \, \chi^2(P_n\|P)
+ \tfrac16 \, \chi^2(P\|P_n)}{\tfrac1{\log \mathrm{e}} \, D(P_n \| P)} = 1.
\end{align}

\item For all $\theta \in (0,1)$,
\begin{align} \label{LB-RE-chi2}
D(P\|Q) \geq (1-\theta) \, \log \biggl(\frac1{1-\theta}\biggr) \; D_{\phi_\theta}(P\|Q).
\end{align}
Moreover, under the assumption in \eqref{lim 1}, for all $\theta \in [0,1]$
\begin{align} \label{29032020b1}
\lim_{n \to \infty}  \frac{D(P\|P_n)}{D_{\phi_\theta}(P\|P_n)} = \tfrac12 \log \mathrm{e}.
\end{align}

\item \cite[Theorem~2]{ISSV16}:
\begin{align} \label{eq: ISSV16}
\tfrac1{\log \mathrm{e}} \, D(P\|Q) \leq \tfrac12 \, \chi^2(P\|Q) + \tfrac14 \, |P-Q|.
\end{align}
\end{enumerate}
\end{proposition}
\begin{proof}
See Section~\ref{subsection: proof f-div ineq.}.
\end{proof}

\begin{remark}
Inequality \eqref{RD2-chi2 new} is locally tight in the sense that \eqref{lim 1} yields
\eqref{locally tight}. This property, however, is not satisfied by \eqref{RD2-chi2}
since the assumption in \eqref{lim 1} implies that
\begin{align} \label{not locally tight}
\lim_{n \to \infty} \frac{\log\bigl(1 + \chi^2(P_n \| P) \bigr)}{D(P_n \| P)} = 2.
\end{align}
\end{remark}

\begin{remark}
Inequality \eqref{RD2-chi2 new} readily yields
\begin{align}
& D(P\|Q) + D(Q\|P) \leq \tfrac12 \left( \chi^2(P\|Q) + \chi^2(Q\|P) \right) \, \log \mathrm{e},
\end{align}
which is proved by a different approach in \cite[Proposition~4]{Simic08}. It is
further shown in \cite[Theorem~2~b)]{ISSV16} that
\begin{align}
\label{eq: symmetrized RE-chi^2}
\sup \frac{D(P \| Q) + D(Q \| P)}{\chi^2(P \| Q) + \chi^2(Q \| P)} = \tfrac12 \log \mathrm{e},
\end{align}
where the supremum is over $P \ll \gg Q$ and $P \neq Q$.
\end{remark}

\subsection{Implications of Theorem~\ref{thm: KL and chi^2}}
\label{subsec: Implications of thm: KL and chi^2}

We next provide two implications of Theorem~\ref{thm: KL and chi^2}.
The first implication, which relies on the Hammersley--Chapman--Robbins (HCR) bound for
the chi-squared divergence (\cite{CR51} and \cite{Hammersley50}), gives
the following tight lower bound on the relative entropy $D(P\|Q)$ as a
function of the means and variances under $P$ and $Q$.

\begin{theorem} \label{theorem: KL LB}
Let $P$ and $Q$ be probability measures defined on the measurable space
$(\Reals, \mathscr{B})$, where $\Reals$ is the real line and $\mathscr{B}$
is the Borel $\sigma$-algebra of subsets of $\Reals$.
Let $m_P$, $m_Q$, $\sigma_P^2$ and $\sigma_Q^2$ denote the expected values
and variances of $X \sim P$ and $Y \sim Q$, i.e.,
\begin{align} \label{eq: 1st and 2nd moments}
& \expectation[X] =: m_P, \; \expectation[Y] =: m_Q,
\quad \Var(X) =: \sigma_P^2, \;  \Var(Y) =: \sigma_Q^2.
\end{align}
\begin{enumerate}[a)]
\item If $m_P \neq m_Q$, then
\begin{align} \label{LB-KL}
D(P\|Q) \geq d(r\|s),
\end{align}
where $d(r\|s) := r \log \frac{r}{s} + (1-r) \log \frac{1-r}{1-s}$,
for $r,s \in [0,1]$, denotes the binary relative entropy (with the
convention that $0 \log \frac{0}{0} = 0$), and
\begin{align}
\label{r}
& r := \frac12 + \frac{b}{4av} \in [0,1], \\[0.1cm]
\label{s}
& s := r - \frac{a}{2v} \in [0,1], \\
\label{a}
& a:= m_P - m_Q, \\
\label{b}
& b:= a^2 + \sigma_Q^2 - \sigma_P^2, \\
\label{v}
& v:= \sqrt{\sigma_P^2 + \frac{b^2}{4a^2}}.
\end{align}

\item
The lower bound in the right side of \eqref{LB-KL} is attained for $P$ and
$Q$ which are defined on the two-element set $\set{U} := \{u_1, u_2\}$, and
\begin{align} \label{PMFs}
P(u_1) = r, \quad Q(u_1) = s,
\end{align}
with $r$ and $s$ in \eqref{r} and \eqref{s}, respectively, and for $m_P \neq m_Q$
\begin{align} \label{u_1,2}
& u_1 := m_P + \sqrt{\frac{(1-r) \sigma_P^2}{r}},
\quad u_2 := m_P - \sqrt{\frac{r \sigma_P^2}{1-r}}.
\end{align}

\item
If $m_P = m_Q$, then

\begin{align}  \label{20200502a1}
\inf_{P,Q} \, D(P\|Q) = 0,
\end{align}
where the infimum in the left side of \eqref{20200502a1} is taken over all $P$ and $Q$ which
satisfy \eqref{eq: 1st and 2nd moments}.
\end{enumerate}
\end{theorem}
\begin{proof}
See Section~\ref{subsection: proof KL-LB}.
\end{proof}

\begin{remark}
Consider the case of the non-equal means in Items~(a) and~(b) of
Theorem~\ref{theorem: KL LB}. If these means are fixed, then
the infimum of $D(P\|Q)$ is zero by choosing arbitrarily large
equal variances. Suppose now that the non-equal means $m_P$ and
$m_Q$ are fixed, as well as one of the variances (either $\sigma_P^2$
or $\sigma_Q^2$).
Numerical experimentation shows that in this case, the achievable lower
bound in \eqref{LB-KL} is monotonically decreasing as a function of the
other variance, and it tends to zero as we let the free
variance tend to infinity. This asymptotic convergence
to zero can be justified by assuming, for example, that $m_P, m_Q$
and $\sigma_Q^2$ are fixed, and $m_P > m_Q$ (the other cases can be
justified in a similar way). Then, it can be verified from
\eqref{r}--\eqref{v} that
\begin{align}
\label{20200507}
r = \frac{(m_P-m_Q)^2}{\sigma_P^2} + O\biggl(\frac1{\sigma_P^4}\biggr),
\quad s = O\biggl(\frac1{\sigma_P^4}\biggr),
\end{align}
which implies that $d(r\|s) \to 0$ as we let $\sigma_P \to \infty$.
The infimum of the relative entropy $D(P\|Q)$ is therefore equal to
zero since the probability measures $P$ and $Q$ in \eqref{PMFs}
and \eqref{u_1,2}, which are defined on a two-element
set and attain the lower bound on the relative entropy under the
constraints in \eqref{eq: 1st and 2nd moments}, have a vanishing relative
entropy in this asymptotic case.
\end{remark}

\begin{remark}
The proof of Item~c) in Theorem~\ref{theorem: KL LB}
suggests explicit constructions of sequences of pairs probability measures $\{(P_n, Q_n)\}$ such that
\begin{enumerate}[a)]
\item The means under $P_n$ and $Q_n$ are both equal to $m$ (independently of $n$);
\item The variance under $P_n$ is equal to $\sigma_P^2$, and the variance under $Q_n$
is equal to $\sigma_Q^2$ (independently of $n$);
\item The relative entropy $D(P_n \| Q_n)$ vanishes as we let $n \to \infty$.
\end{enumerate}
This yields in particular \eqref{20200502a1}.
\end{remark}

\vspace*{0.2cm}
A second consequence of Theorem~\ref{thm: KL and chi^2} gives the
following result. Its first part holds due to the concavity of
$\exp\bigl(-D(P\|\cdot)\bigr)$ (see \cite[Problem~4.2]{Verdu20}).
The second part is new, and its proof relies on Theorem~\ref{thm: KL and chi^2}.
As an educational note, we provide an alternative proof of the first part
by relying on Theorem~\ref{thm: KL and chi^2}.
\begin{theorem} \label{thm: F}
Let $P \ll Q$, and $F \colon [0,1] \to [0, \infty)$ be given by
\begin{align} \label{def:F}
F(\lambda) := D \bigl(P \, \| \, (1-\lambda)P + \lambda Q \bigr),
\quad \forall \, \lambda \in [0,1].
\end{align}
Then, for all $\lambda \in [0,1]$,
\begin{align}
\label{UB on F}
F(\lambda) \leq \log \Biggl( \frac1{1-\lambda+\lambda \exp\bigl(-D(P\|Q)\bigr)}
\Biggr),
\end{align}
with an equality if $\lambda=0$ or $\lambda=1$. Moreover, $F$ is
monotonically increasing, differentiable, and it satisfies
\begin{align}
\label{diff F 2a}
& F'(\lambda) \geq \frac1{\lambda} \Bigl[\exp\bigl(F(\lambda)\bigr)
- 1 \Bigr] \log \mathrm{e}, \quad \forall \, \lambda \in (0,1], \\[0.2cm]
\label{diff F 2b}
& \lim_{\lambda \to 0^+} \frac{F'(\lambda)}{\lambda}
= \chi^2(Q \| P) \, \log \mathrm{e},
\end{align}
so, the limit in \eqref{diff F 2b} is twice larger than the value of the lower
bound on this limit as it follows from the right side of \eqref{diff F 2a}.
\end{theorem}
\begin{proof}
See Section~\ref{subsection: proof of thm: F}.
\end{proof}

\begin{remark}
By the convexity of the relative entropy, it follows that
$F(\lambda) \leq \lambda \, D(P\|Q)$ for all $\lambda \in [0,1]$.
It can be verified, however, that the inequality
$1-\lambda + \lambda \exp(-x) \geq \exp(-\lambda x)$ holds for
all $x \geq 0$ and $\lambda \in [0,1]$. Letting $x:= D(P\|Q)$
implies that the upper bound on $F(\lambda)$ in the right side
of \eqref{UB on F} is tighter than or equal to the latter bound
(with an equality if and only if either $\lambda \in \{0,1\}$
or $P \equiv Q$).
\end{remark}

\begin{corollary}
\label{corollary: UB on relative entropy}
Let $\{P_j\}_{j=1}^m$, with $m \in \naturals$, be probability measures
defined on a measurable space $(\set{X}, \mathscr{F})$, and let
$\{\alpha_j\}_{j=1}^m$ be a sequence of non-negative numbers that sum
to~1. Then, for all $i \in \{1, \ldots, m\}$,
\begin{align}
\label{1911a1}
D\Biggl( P_i \, \| \, \sum_{j=1}^m \alpha_j P_j \Biggr) & \leq
-\log \Biggl( \alpha_i + (1-\alpha_i)
\exp\biggl(-\tfrac1{1-\alpha_i} \sum_{j \neq i} \alpha_j \,
D(P_i \| P_j) \biggr) \Biggr).
\end{align}
\end{corollary}
\begin{proof}
For an arbitrary $i \in \{1, \ldots, m\}$, apply the upper bound in
the right side of \eqref{UB on F} with $\lambda := 1-\alpha_i$,
$P := P_i$ and
$Q := \tfrac1{1-\alpha_i} \, \underset{j \neq i}{\sum} \alpha_j P_j$.
The right side of \eqref{1911a1} is obtained from \eqref{UB on F} by
invoking the convexity of the relative entropy, which gives
$D(P_i \|Q) \leq \tfrac1{1-\alpha_i} \underset{j \neq i}{\sum} \alpha_j D(P_i \| P_j)$.
\end{proof}

The next result provides an upper bound on the non-negative difference between
the entropy of a convex combination of distributions and the respective convex
combination of the individual entropies (it is also termed as the concavity
deficit of the entropy function in \cite[Section~3]{Allerton19}).

\begin{corollary}
\label{corollary: entropy of a mixture of distributions}
Let $\{P_j\}_{j=1}^m$, with $m \in \naturals$, be probability measures
defined on a measurable space $(\set{X}, \mathscr{F})$, and let
$\{\alpha_j\}_{j=1}^m$ be a sequence of non-negative numbers that sum
to~1. Then, the entropy of the mixed distribution $\sum_j \alpha_j P_j$ satisfies
\begin{align}
0 & \leq H \Biggl( \sum_{j=1}^m \alpha_j P_j \Biggr) - \sum_{j=1}^m \alpha_j H(P_j) \nonumber \\
\label{eq: entropy of a mixture of distributions}
& \leq - \sum_{i=1}^m \alpha_i \log \Biggl( \alpha_i + (1-\alpha_i)
\exp\biggl(-\tfrac1{1-\alpha_i} \sum_{j \neq i} \alpha_j \,
D(P_i \| P_j) \biggr) \Biggr).
\end{align}
\end{corollary}
\begin{proof}
The lower bound holds due to the concavity of the entropy function.
The upper bound readily follows from Corollary~\ref{corollary: UB on relative entropy},
and the identity
\begin{align}  \label{eq: identity mixtures}
H \Biggl( \sum_{j=1}^m \alpha_j P_j \Biggr) - \sum_{j=1}^m \alpha_j H(P_j)
= \sum_{i=1}^m \alpha_i D\Biggl(P_i \, \| \, \sum_{j=1}^m \alpha_j P_j \Biggr).
\end{align}
\end{proof}
\begin{remark}
The upper bound in \eqref{eq: entropy of a mixture of distributions} refines
the known bound (see, e.g., \cite[Lemma~2.2]{WangM_IT14})
\begin{align}
\label{eq2: entropy of a mixture of distributions}
H \Biggl( \sum_{j=1}^m \alpha_j P_j \Biggr) - \sum_{j=1}^m \alpha_j H(P_j)
\leq \sum_{j=1}^m \alpha_j \, \log \frac1{\alpha_j} = H(\underline{\alpha}),
\end{align}
by relying on all the $\tfrac12 m (m-1)$ pairwise relative entropies between the
individual distributions $\{P_j\}_{j=1}^m$. Another refinement of
\eqref{eq2: entropy of a mixture of distributions}, expressed in terms of total
variation distances, has been recently provided in \cite[Theorem~3.1]{Allerton19}.
\end{remark}

\subsection{Monotonic Sequences of $f$-divergences and an Extension
of Theorem~\ref{thm: KL and chi^2}}
\label{subsec: monotonic_sequence}

The present subsection generalizes Theorem~\ref{thm: KL and chi^2},
and it also provides relations between $f$-divergences
which are defined in a recursive way.

\begin{theorem} \label{theorem: monotonic_f_divergence_seq}
Let $P$ and $Q$ be probability measures defined on a measurable space
$(\set{X}, \mathscr{F})$. Let $R_\lambda$, for $\lambda \in [0,1]$, be
the convex combination of $P$ and $Q$ as in \eqref{mixture PMF}.
Let $f_0 \colon (0, \infty) \to \Reals$ be a convex function
with $f_0(1)=0$, and let $\{f_k(\cdot)\}_{k=0}^{\infty}$ be a sequence
of functions that are defined on $(0, \infty)$ by the recursive equation
\begin{align} \label{eq:recursion f_k}
f_{k+1}(x) := \int_0^{1-x} f_k(1-s) \; \frac{\mathrm{d}s}{s}, \quad x >0,
\; \; k \in \{0,1,\ldots\}.
\end{align}
Then,
\begin{enumerate}[1)]
\item $\bigl\{D_{f_k}(P\|Q)\bigr\}_{k=0}^{\infty}$ is a non-increasing
(and non-negative) sequence of $f$-divergences.
\item For all $\lambda \in[0,1]$ and $k \in \{0, 1, \ldots\}$,
\begin{align} \label{eq:fD_integral}
D_{f_{k+1}}(R_\lambda \| P)
= \int_0^\lambda D_{f_k}(R_s \| P)  \; \frac{\mathrm{d}s}{s}.
\end{align}
\end{enumerate}
\end{theorem}

\begin{proof}
See Section~\ref{subsection: proof monotonic_f_divergence_seq}.
\end{proof}

We next use the polylogarithm functions, which satisfy the recursive
equation \cite[Eq.~(7.2)]{Lewin}:
\begin{align} \label{eq:polylog}
\mathrm{Li}_{k}(x) :=
\begin{dcases}
\frac{x}{1-x},  & \quad \mbox{if} \hspace*{0.15cm} k=0, \\[0.2cm]
\int_0^{x} \frac{\mathrm{Li}_{k-1}(s)}{s} \; \mathrm{d}s,
& \quad \mbox{if} \hspace*{0.15cm} k\geq 1.
\end{dcases}
\end{align}
This gives $\mathrm{Li}_1(x) = -\log_{\mathrm{e}}(1-x)$,
$\mathrm{Li}_2(x) = -\int_0^{x} \frac1{s} \, \log_{\mathrm{e}}(1-s) \, \mathrm{d}s$
and so on, which are real-valued and finite for $x<1$.

\begin{corollary} \label{corollary: polylogarithm}
Let
\begin{align} \label{eq: def f_k}
f_k(x) := \mathrm{Li}_{k}(1-x), \quad x > 0, \; \; k \in \{0, 1, \ldots\}.
\end{align}
Then, \eqref{eq:fD_integral} holds for all $\lambda \in[0,1]$ and
$k \in \{0, 1, \ldots\}$. Furthermore, setting $k=0$ in \eqref{eq:fD_integral}
yields \eqref{identity: KL} as a special case.
\end{corollary}

\begin{proof}
See Section~\ref{proof: corollary-polylogarithm}.
\end{proof}

\subsection{On Probabilities and $f$-divergences}
\label{subsec: prob. and f-divergences}

The following result relates probabilities of sets to $f$-divergences.
\begin{theorem} \label{thm: gen. Csiszar}
Let $(\set{X}, \mathscr{F}, \mu)$ be a probability space, and let
$\set{C} \in \mathscr{F}$ be a measurable set with $\mu(\set{C}) > 0$.
Define the conditional probability measure
\begin{align} \label{cond. PM}
\mu_{\set{C}}(\set{E}) := \frac{\mu(\set{C} \cap \set{E})}{\mu(\set{C})},
\quad \forall \, \set{E} \in \mathscr{F}.
\end{align}
Let $f \colon (0, \infty) \to \Reals$ be an arbitrary convex function
with $f(1)=0$, and assume (by continuous extension of $f$ at zero) that
$f(0) := \underset{t \to 0^+}{\lim} f(t) < \infty$. Furthermore, let
$\widetilde{f} \colon (0, \infty) \to \Reals$ be the convex function
which is given by
\begin{align} \label{f conjugate}
\widetilde{f}(t) := t f\biggl(\frac1t\biggr), \quad \forall \, t>0.
\end{align}
Then,
\begin{align} \label{eq: gen. Csiszar}
D_f(\mu_{\set{C}} \| \mu)
= \widetilde{f}\bigl(\mu(\set{C})\bigr) + \bigl(1-\mu(\set{C})\bigr) \, f(0).
\end{align}
\end{theorem}
\begin{proof}
See Section~\ref{subsection: proof of Csiszar84 extended}.
\end{proof}

Connections of probabilities to the relative entropy, and to the
chi-squared divergence, are next exemplified as special cases of
Theorem~\ref{thm: gen. Csiszar}.
\begin{corollary}
\label{corollary: gen. Csiszar}
In the setting of Theorem~\ref{thm: gen. Csiszar},
\begin{align}
\label{eq: Csiszar84a}
& D\bigl( \mu_{\set{C}} \| \mu \bigr) = \log \frac1{\mu\bigl(\set{C}\bigr)}, \\[0.1cm]
\label{eq: Csiszar84b}
& \chi^2\bigl( \mu_{\set{C}} \| \mu \bigr) = \frac1{\mu\bigl(\set{C}\bigr)} - 1,
\end{align}
so \eqref{RD2-chi2} is satisfied in this case with equality.
More generally, for all $\alpha \in (0, \infty)$,
\begin{align}
\label{eq: Csiszar84c}
D_{\alpha}\bigl( \mu_{\set{C}} \| \mu \bigr) = \log \frac1{\mu\bigl(\set{C}\bigr)}.
\end{align}
\end{corollary}
\begin{proof}
See Section~\ref{subsection: proof of Csiszar84 extended}.
\end{proof}

\begin{remark}
In spite of its simplicity, \eqref{eq: Csiszar84a} proved very useful
in the seminal work by Marton on transportation-cost inequalities,
proving concentration of measures by information-theoretic tools
\cite{Marton1, Marton2} (see also \cite[Chapter~8]{BoucheronLM}
and \cite[Chapter~3]{RS_FnT}). As a side note, the simple identity
\eqref{eq: Csiszar84a} was apparently first explicitly used by Csisz\'{a}r
(see \cite[Eq.~(4.13)]{Csiszar84}).
\end{remark}

\section{Applications}
\label{section: applications}

This section provides applications of our results in
Section~\ref{section: main results}. These include universal lossless
compression, method of types and large deviations, and strong
data-processing inequalities (SDPIs).

\subsection{Application of Corollary~\ref{corollary: UB on relative entropy}:
Shannon Code for Universal Lossless Compression}
\label{subsec: UB on RE - Poisson}

Consider $m>1$ discrete, memoryless and stationary sources with
probability mass functions $\{P_i\}_{i=1}^m$,
and assume that the symbols are emitted by one of these sources
with an {\em a priori} probability $\alpha_i$ for source no.~$i$,
where $\{\alpha_i\}_{i=1}^m$ are positive and sum to~1.

For lossless data compression by a universal source code, suppose
that a single source code is designed with respect to the average
probability mass function
$P := \overset{m}{\underset{j=1}{\sum}} \alpha_j P_j$.

Assume that the designer uses a Shannon code, where the code
assignment for a symbol $x \in \set{X}$ is of length
$\ell(x) = \Bigl\lceil \log \frac1{P(x)}\Bigl\rceil$ bits
(logarithms are on base~2). Due to the mismatch in the source
distribution, the average codeword length $\ell_{\mathrm{avg}}$
satisfies (see \cite[Proposition~3.B]{ClarkeB90})
\begin{align} \label{mismatch}
\sum_{i=1}^m \alpha_i H(P_i) + \sum_{i=1}^m \alpha_i D(P_i \| P)
\leq \ell_{\mathrm{avg}} \leq \sum_{i=1}^m \alpha_i H(P_i)
+ \sum_{i=1}^m \alpha_i D(P_i \| P) + 1.
\end{align}
The fractional penalty in the average codeword length, denoted by
$\nu$, is defined to be equal to the ratio of the penalty in
the average codeword length as a result of the source mismatch, and
the average codeword length in case of a perfect matching. From
\eqref{mismatch}, it follows that
\begin{align}
\label{bounds on nu}
\frac{\overset{m}{\underset{i=1}{\sum}} \alpha_i \, D(P_i \| P)}{1
+ \overset{m}{\underset{i=1}{\sum}} \alpha_i H(P_i)}
\leq \nu \leq \frac{1 + \overset{m}{\underset{i=1}{\sum}} \alpha_i \,
D(P_i \| P)}{\overset{m}{\underset{i=1}{\sum}} \alpha_i H(P_i)}.
\end{align}

We next rely on Corollary~\ref{corollary: UB on relative entropy}
to obtain an upper bound on $\nu$ which is expressed as a function
of the $m(m-1)$ relative entropies $D(P_i \| P_j)$ for all $i \neq j$ in
$\{1, \ldots, m\}$. This is useful if, e.g., the $m$
relative entropies in the left and right sides of \eqref{bounds on nu}
do not admit closed-form expressions, in contrast to the $m(m-1)$
relative entropies $D(P_i \| P_j)$ for $i \neq j$. We next exemplify
this case.

For $i \in \{1, \ldots, m\}$, let $P_i$ be a Poisson
distribution with parameter $\lambda_i > 0$.
Consequently, for $i, j \in \{1, \ldots, m\}$, the relative entropy
from $P_i$ to $P_j$ admits the closed-form expression
\begin{align}
\label{1911b3}
& D(P_i \| P_j)
= \lambda_i \, \log \biggl(\frac{\lambda_i}{\lambda_j}\biggr)
+ (\lambda_j - \lambda_i) \log \mathrm{e}.
\end{align}
From \eqref{1911a1} and \eqref{1911b3}, it
follows that
\begin{align}
\label{UB - Poisson}
D(P_i \| P) & \leq -\log \biggl( \alpha_i + (1-\alpha_i)
\exp \biggl(-\frac{f_i(\underline{\alpha},
\underline{\lambda})}{1-\alpha_i} \biggr) \biggr),
\end{align}
where
\begin{align}
f_i(\underline{\alpha}, \underline{\lambda})
&:= \sum_{j \neq i} \alpha_j \, D(P_i \| P_j) \\
&=\sum_{j \neq i}
\biggl\{ \alpha_j \biggl[ \lambda_i \, \log \biggl(\frac{\lambda_i}{\lambda_j}\biggr)
+ (\lambda_j - \lambda_i) \log \mathrm{e} \biggr] \biggr\}.
\end{align}
The entropy of a Poisson distribution, with parameter $\lambda_i > 0$,
is given by the following integral representation (\cite{EvansB88,Knessl98,MerhavS20}):
\begin{align}
\label{entropy Poisson}
H(P_i) = \lambda_i \log \biggl( \frac{\mathrm{e}}{\lambda_i} \biggr) +
\int_0^\infty \biggl(\lambda_i - \frac{1-\mathrm{e}^{-\lambda_i
(1-\mathrm{e}^{-u})}}{1-\mathrm{e}^{-u}}\biggr) \frac{\mathrm{e}^{-u}}{u}
\, \mathrm{d}u \; \log \mathrm{e}.
\end{align}
Combining \eqref{bounds on nu}, \eqref{UB - Poisson} and
\eqref{entropy Poisson} finally gives an upper bound on
$\nu$ in the considered setup.

\begin{example}
Consider five discrete memoryless sources where the probability mass
function of source no. $i$ is given by $P_i = \mathrm{Poisson}(\lambda_i)$ with
$\underline{\lambda} = [16, 20, 24, 28, 32]$. Suppose that the symbols are
emitted from one of the sources with equal probability, so
$\underline{\alpha} = \bigl[\tfrac15, \tfrac15, \tfrac15, \tfrac15, \tfrac15 \bigr]$.
Let $P := \tfrac15 (P_1 + \ldots + P_5)$
be the average probability mass function of the five sources. The term
$\sum_i \alpha_i \, D(P_i \| P)$, which appears in the
numerators of the upper and lower bounds on $\nu$ (see \eqref{bounds on nu})
does not lend itself to a closed-form expression, and it is not even an easy task
to calculate it numerically due to the need to compute an infinite series which involves
factorials. We therefore apply the closed-form upper bound in \eqref{UB - Poisson}
to get that $\sum_i \alpha_i \, D(P_i \| P) \leq 1.46$~bits, whereas the
upper bound which follows from the convexity of the relative entropy (i.e.,
$\sum_i \alpha_i f_i(\underline{\alpha}, \underline{\lambda})$) is equal
to 1.99~bits (both upper bounds are smaller than the trivial bound
$\log_2 5 \approx 2.32$~bits).
From \eqref{bounds on nu}, \eqref{entropy Poisson} and the stronger upper
bound on $\sum_i \alpha_i \, D(P_i \| P)$, the improved upper bound on $\nu$
is equal to~$57.0\%$ (as compared to a looser upper bound of $69.3\%$, which
follows from \eqref{bounds on nu}, \eqref{entropy Poisson} and the looser
upper bound on $\sum_i \alpha_i \, D(P_i \| P)$ that is equal to~1.99 bits).
\end{example}

\subsection{Application of Theorem~\ref{theorem: KL LB} in the Context
of the Method of Types and Large Deviations Theory}
\label{subsec: method of types}

Let $X^n = (X_1, \ldots, X_n)$ be a sequence of i.i.d. random variables with
$X_1 \sim Q$ where $Q$ is a probability measure defined on a
finite set $\set{X}$, and $Q(x) > 0$ for all $x \in \set{X}$. Let
$\mathscr{P}$ be a set of probability measures on $\set{X}$ such
that $Q \notin \mathscr{P}$, and suppose that the closure of $\mathscr{P}$
coincides with the closure of its interior. Then, by Sanov's theorem
(see, e.g., \cite[Theorem~11.4.1]{Cover_Thomas} and \cite[Theorem~3.3]{Csiszar98}),
the probability that the empirical distribution $\widehat{P}_{X^n}$
belongs to $\mathscr{P}$ vanishes exponentially at the rate
\begin{align} \label{Sanov}
\lim_{n \to \infty} \frac1n \, \log \frac1{\prob[\widehat{P}_{X^n} \in \mathscr{P}]}
= \inf_{P \in \mathscr{P}} D(P \|Q).
\end{align}
Furthermore, for finite $n$, the method of types yields the following upper bound
on this rare event
\begin{align}  \label{method of types1}
\prob[\widehat{P}_{X^n} \in \mathscr{P}] & \leq \binom{n+|\set{X}|-1}{|\set{X}|-1}
\exp\Bigl(-n \inf_{P \in \mathscr{P}} D(P \|Q) \Bigr) \\
\label{method of types2}
& \leq (n+1)^{|\set{X}|-1} \, \exp\Bigl(-n \inf_{P \in \mathscr{P}} D(P \|Q) \Bigr),
\end{align}
whose exponential decay rate coincides with the exact asymptotic result in \eqref{Sanov}.

Suppose that $Q$ is not fully known, but its mean $m_Q$ and variance $\sigma_Q^2$
are available. Let $m_1 \in \Reals$ and $\delta_1, \varepsilon_1, \sigma_1 > 0$ be
fixed, and let $\mathscr{P}$ be the set of all probability measures $P$, defined
on the finite set $\set{X}$, with mean $m_P \in [m_1 - \delta_1, m_1 + \delta_1]$
and variance $\sigma_P^2 \in [\sigma_1^2 - \varepsilon_1, \sigma_1^2 + \varepsilon_1]$
where $|m_1 - m_Q| > \delta_1$. Hence, $\mathscr{P}$ coincides with the closure of
its interior, and $Q \notin \mathscr{P}$.

The lower bound on the relative entropy in Theorem~\ref{theorem: KL LB}, used in
conjunction with the upper bound in \eqref{method of types2}, can serve to obtain
an upper bound on the probability of the event that the empirical distribution of
$X^n$ belongs to the set $\mathscr{P}$, regardless of the uncertainty in $Q$.
This gives
\begin{align}  \label{method of types3}
\prob[\widehat{P}_{X^n} \in \mathscr{P}] & \leq
(n+1)^{|\set{X}|-1} \, \exp\bigl(-n d^\ast\bigr),
\end{align}
where
\begin{align} \label{d star}
d^\ast := \inf_{m_P, \sigma_P^2} d(r\|s),
\end{align}
and, for fixed $(m_P, m_Q, \sigma_P^2, \sigma_Q^2)$, the parameters $r$ and $s$ are
given in \eqref{r} and \eqref{s}, respectively.

Standard algebraic manipulations that rely on \eqref{method of types3} lead to the following
result, which is expressed as a function of the Lambert $W$ function \cite{CGHJK96}.
This function, which finds applications in various engineering and scientific
fields, is a standard built-in function in mathematical software tools
such as Mathematica, Matlab and Maple. Applications of the Lambert $W$ function
in information theory and coding are briefly surveyed in \cite{ITA14}.
\begin{proposition}
\label{prop: W function}
For $\varepsilon \in (0,1)$, let $n^\ast := n^\ast(\varepsilon)$ denote the
minimal value of $n \in \naturals$ such that the upper bound
in the right side of \eqref{method of types3} does not exceed $\varepsilon \in (0,1)$.
Then, $n^\ast$ admits the following closed-form expression:
\begin{align} \label{eq: minimal n}
n^\ast = \max \Biggl\{ \left\lceil -\frac{\bigl(|\set{X}|-1\bigr) \, W_{-1}(\eta)
\, \log \mathrm{e}}{d^\ast} \right\rceil - 1, \, 1 \Biggr\},
\end{align}
with
\begin{align}
& \eta := - \, \frac{d^\ast \, \bigl( \varepsilon \,
\exp(-d^\ast) \bigr)^{1/(|\set{X}|-1)}}{\bigl(|\set{X}|-1\bigr)
\log \mathrm{e}} \in \bigl[-\tfrac1{\mathrm{e}}, 0),
\end{align}
and $W_{-1}(\cdot)$ in the right side of \eqref{eq: minimal n} denotes the secondary real-valued branch of the
Lambert $W$ function (i.e., $x := W_{-1}(y)$ where
$W_{-1} \colon \bigl[-\tfrac1{\mathrm{e}}, 0) \to (-\infty, -1]$ is the inverse function of $y := x \mathrm{e}^x$).
\end{proposition}

\begin{example}
Let $Q$ be an arbitrary probability measure, defined on a finite set $\set{X}$, with mean
$m_Q = 40$ and variance $\sigma_Q^2 = 20$. Let $\mathscr{P}$ be the set of all probability
measures $P$, defined on $\set{X}$, whose mean $m_P$ and variance $\sigma_P^2$ lie in the
intervals $[43, 47]$ and $[18, 22]$, respectively.
Suppose that it is required that, for all probability measures $Q$ as above, the probability
that the empirical distribution of the i.i.d. sequence $X^n \sim Q^n$ be included in the set
$\mathscr{P}$ is at most $\varepsilon = 10^{-10}$. We rely here on the upper bound in
\eqref{method of types3}, and impose the stronger condition where it should not exceed $\varepsilon$.
By this approach, it is obtained numerically from \eqref{d star} that $d^\ast = 0.203$ nats.
We next examine two cases:
\begin{enumerate}[i)]
\item If $|\set{X}|=2$, then it follows from \eqref{eq: minimal n} that $n^\ast = 138$.
\item Consider a richer alphabet size of the i.i.d. samples where, e.g., $|\set{X}|=100$.
By relying on the same universal lower bound $d^\ast$, which holds independently of
the value of $|\set{X}|$ ($\set{X}$ can possibly be an infinite set), it follows from
\eqref{eq: minimal n} that $n^\ast = 4170$ is the minimal value such that the upper bound
in \eqref{method of types3} does not exceed $10^{-10}$.
\end{enumerate}
\end{example}

We close this discussion by providing numerical experimentation of the lower bound on
the relative entropy in Theorem~\ref{theorem: KL LB}, and comparing this attainable
lower bound (see Item~b) of Theorem~\ref{theorem: KL LB}) with the following closed-form
expressions for relative entropies:
\begin{enumerate}[a)]
\item The relative entropy between real-valued Gaussian distributions is given by
\begin{align} \label{eq: KL Gaussians}
D\bigl( \mathcal{N}(m_P, \sigma_P^2) \, \| \, \mathcal{N}(m_Q, \sigma_Q^2) \bigr)
= \log \frac{\sigma_Q}{\sigma_P}
+ \frac12 \biggl[\frac{(m_P-m_Q)^2+\sigma_P^2}{\sigma_Q^2} - 1 \biggr] \, \log \mathrm{e}.
\end{align}
\item Let $E_\mu$ denote a random variable which is exponentially distributed with
mean $\mu > 0$; its probability density function is given by
\begin{align} \label{exp pdf}
e_\mu(x) = \frac1{\mu} \, \mathrm{e}^{-x/\mu} \; 1\{x \geq 0\}.
\end{align}
Then, for $a_1, a_2 >0$ and $d_1, d_2 \in \Reals$,
\begin{align} \label{eq: KL exponentials}
D(E_{a_1} + d_1 \| E_{a_2} + d_2) =
\begin{dcases}
\log \frac{a_2}{a_1} + \frac{d_1 + a_1 - d_2 - a_2}{a_2}
\; \log \mathrm{e}, & \quad d_1 \geq d_2, \\
\infty, & \quad d_1 < d_2.
\end{dcases}
\end{align}
In this case, the means under $P$ and $Q$ are $m_P = d_1 + a_1$ and $m_Q = d_2 + a_2$,
respectively, and the variances are $\sigma_P^2 = a_1^2$ and $\sigma_Q^2 = a_2^2$.
Hence, for obtaining the required means and variances, set
\begin{align}
a_1 = \sigma_P, \quad a_2 = \sigma_Q, \quad d_1 = m_P - \sigma_P, \quad d_2 = m_Q - \sigma_Q.
\end{align}
\end{enumerate}

\begin{example}
We compare numerically the attainable lower bound on the relative entropy, as it is given in
\eqref{LB-KL}, with the two relative entropies in \eqref{eq: KL Gaussians} and
\eqref{eq: KL exponentials}.
\begin{enumerate}[i)]
\item If $(m_P, m_Q, \sigma_P^2, \sigma_Q^2) = (45, 40, 20, 20)$, then the lower bound
in \eqref{LB-KL} is equal to 0.521~nats, and the two relative entropies in \eqref{eq: KL Gaussians}
and \eqref{eq: KL exponentials} are equal to 0.625 and 1.118~nats, respectively.
\item If $(m_P, m_Q, \sigma_P^2, \sigma_Q^2) = (50, 35, 10, 20)$, then the lower bound in
\eqref{LB-KL} is equal to 2.332~nats, and the two relative entropies in \eqref{eq: KL Gaussians}
and \eqref{eq: KL exponentials} are equal to 5.722 and 3.701~nats, respectively.
\end{enumerate}
\end{example}

\subsection{Strong Data-Processing Inequalities and Maximal Correlation}
\label{subsec: SDPI}

The information contraction is a fundamental concept in information
theory. The contraction of $f$-divergences through channels is captured
by data-processing inequalities, which can be further tightened by the
derivation of SDPIs with channel-dependent or source-channel dependent
contraction coefficients (see, e.g.,
\cite{CohenKZ98, Cohen93, MakurP18, MakurZ15, Makur_PhD19, PolyanskiyW17, Raginsky16, Sason19}).

We next provide necessary definitions which are relevant for the presentation
in this subsection.
\begin{definition}
\label{def: contraction}
Let $Q_X$ be a probability distribution which is defined on a set $\set{X}$,
and that is not a point mass, and let $W_{Y|X} \colon \set{X} \to \set{Y}$
be a stochastic transformation. The {\em contraction coefficient for
$f$-divergences} is defined as
\begin{align}
\label{contraction coef.}
\mu_f(Q_X, W_{Y|X}) := \underset{P_X: \, D_f(P_X \| Q_X) \in (0, \infty)}{\sup} \,
\frac{D_f(P_Y \| Q_Y)}{D_f(P_X \| Q_X)},
\end{align}
where, for all $y \in \set{Y}$,
\begin{align}
\label{23062019a1}
& P_Y(y) = (P_X W_{Y|X}) \: (y) := \int_{\set{X}} \mathrm{d}P_X(x) \, W_{Y|X}(y|x), \\[0.1cm]
\label{23062019a2}
& Q_Y(y) = (Q_X W_{Y|X}) \: (y) := \int_{\set{X}} \mathrm{d}Q_X(x) \, W_{Y|X}(y|x).
\end{align}
The notation in \eqref{23062019a1} and \eqref{23062019a2} is consistent with the standard
notation used in information theory (see, e.g., the first displayed equation after (3.2)
in \cite{Csiszar_Korner}).
\end{definition}

The derivation of good upper bounds on contraction coefficients for $f$-divergences, which are strictly
smaller than~1, lead to SDPIs. These inequalities find their applications, e.g., in studying the
exponential convergence rate of an irreducible, time-homogeneous and reversible discrete-time Markov
chain to its unique invariant distribution over its state space (see, e.g., \cite[Section~2.4.3]{Makur_PhD19}
and \cite[Section~2]{PolyanskiyW17}). It is in sharp contrast to DPIs which do not yield convergence
to stationarity at any rate. We return to this point later in this subsection, and determine the exact
convergence rate to stationarity under two parametric families of $f$-divergences.

We next rely on Theorem~\ref{thm: KL and chi^2} to obtain upper bounds
on the contraction coefficients for the following $f$-divergences.

\begin{definition} \label{def: parametric f-divergences}
For $\alpha \in (0,1]$, the $\alpha$-skew $K$-divergence is given by
\begin{align}  \label{K div.}
K_\alpha(P\|Q) := D\bigl(P \, \| \, (1-\alpha)P + \alpha Q\bigr),
\end{align}
and, for $\alpha \in [0,1]$, let
\begin{align}  \label{S div.}
S_\alpha(P\|Q) &:= \alpha \, D\bigl(P \, \| \, (1-\alpha)P + \alpha Q\bigr)
+ (1-\alpha) \, D\bigl(Q \, \| \, (1-\alpha)P + \alpha Q\bigr) \\
\label{S-K div.}
& \, = \alpha \, K_\alpha(P\|Q) + (1-\alpha) \, K_{1-\alpha}(Q \| P),
\end{align}
with the convention that $K_0(P\|Q) \equiv 0$ (by a continuous extension at $\alpha=0$
in \eqref{K div.}). These divergence measures are specialized to the relative entropies:
\begin{align} \label{eq: 20200417c1}
K_1(P\|Q) = D(P\|Q) = S_1(P\|Q), & \quad S_0(P\|Q) = D(Q\|P),
\end{align}
and $S_{\frac12}(P\|Q)$ is the Jensen--Shannon divergence \cite{BurbeaR82,Lin91,JSD97}
(also known as the capacitory discrimination \cite{Topsoe_IT00}):
\begin{align} \label{JS}
S_{\frac12}(P\|Q) &= \tfrac12 \, D\bigl(P \, \| \, \tfrac12(P+Q)\bigr)
+ \tfrac12 \, D\bigl(Q \, \| \, \tfrac12(P+Q)\bigr) \\
&= H\bigl(\tfrac12(P+Q)\bigr) - \tfrac12 H(P) - \tfrac12 H(Q) := \mathrm{JS}(P\|Q).
\end{align}
It can be verified that the divergence measures in \eqref{K div.} and \eqref{S div.}
are $f$-divergences:
\begin{align}
\label{eq: 20200412a1}
& K_\alpha(P\|Q) = D_{k_\alpha}(P\|Q), \quad \alpha \in (0,1], \\
\label{eq: 20200412a2}
& S_\alpha(P\|Q) = D_{s_\alpha}(P\|Q), \quad \; \alpha \in [0,1],
\end{align}
with
\begin{align}
\label{eq: 20200412b1}
k_\alpha(t) &:= t \log t - t \log \bigl(\alpha + (1-\alpha)t \bigr),
\quad t>0, \; \; \alpha \in (0,1], \\
\label{eq: 20200412b2}
s_\alpha(t) &:= \alpha t \log t - \bigl(\alpha t + 1-\alpha \bigr) \,
\log\bigl(\alpha + (1-\alpha)t\bigr) \\
\label{eq: 20200412b2.2}
& \; = \alpha k_\alpha(t) + (1-\alpha) t \, k_{1-\alpha}\biggl(\frac{1}{t}\biggr),
\quad t>0, \; \; \alpha \in [0,1],
\end{align}
where $k_\alpha(\cdot)$ and $s_\alpha(\cdot)$ are strictly convex functions
on $(0, \infty)$, and vanish at~1.
\end{definition}

\begin{remark}
The $\alpha$-skew $K$-divergence in \eqref{K div.} is considered in \cite{Lin91} and
\cite[(13)]{Nielsen20} (including pointers in the latter paper to its utility). The
divergence in \eqref{S div.} is akin to Lin's measure in \cite[(4.1)]{Lin91}, the
asymmetric $\alpha$-skew Jensen--Shannon divergence in \cite[(11)--(12)]{Nielsen20},
the symmetric $\alpha$-skew Jensen--Shannon divergence in \cite[(16)]{Nielsen20},
and divergence measures in \cite{AsadiEKS19} which involve arithmetic and geometric
means of two probability distributions. Properties and applications of quantum skew
divergences are studied in \cite{Audenaert} and references therein.
\end{remark}

\begin{theorem} \label{thm: contraction coefficients}
The $f$-divergences in \eqref{K div.} and \eqref{S div.} satisfy the following
integral identities, which are expressed in terms of the Gy\"{o}rfi--Vajda divergence
in \eqref{eq: f - GV01}:
\begin{align}
\label{eq: 20200412c1}
& \tfrac1{\log \mathrm{e}} \, K_\alpha(P\|Q) = \int_0^\alpha s D_{\phi_s}(P\|Q) \, \mathrm{d}s,
\hspace*{1cm} \alpha \in (0,1], \\
\label{eq: 20200412c2}
& \tfrac1{\log \mathrm{e}} \, S_\alpha(P\|Q) = \int_0^1 g_\alpha(s) \, D_{\phi_s}(P\|Q) \, \mathrm{d}s,
\quad \alpha \in [0,1],
\end{align}
with
\begin{align}
\label{eq: 20200412c3}
g_\alpha(s) := \alpha s \, 1\bigl\{s \in (0, \alpha]\bigr\} + (1-\alpha)(1-s) \,
1\bigl\{s \in [\alpha, 1) \bigr\}, \quad (\alpha, s) \in [0,1]^2.
\end{align}
Moreover, the contraction coefficients for these $f$-divergences are related as follows:
\begin{align}
\label{eq: 20200412d1}
\mu_{\chi^2}(Q_X, W_{Y|X}) & \leq \mu_{k_\alpha}(Q_X, W_{Y|X}) \leq \sup_{s \in (0, \alpha]}
\mu_{\phi_s}(Q_X, W_{Y|X}), \quad \alpha \in (0,1], \\
\label{eq: 20200412d2}
\mu_{\chi^2}(Q_X, W_{Y|X}) & \leq \mu_{s_\alpha}(Q_X, W_{Y|X}) \leq \sup_{s \in (0, 1)}
\mu_{\phi_s}(Q_X, W_{Y|X}), \quad \; \alpha \in [0,1],
\end{align}
where $\mu_{\chi^2}(Q_X, W_{Y|X})$ denotes the contraction coefficient for the chi-squared
divergence.
\end{theorem}
\begin{proof}
See Section~\ref{subsection: proof of thm: contraction coefficients}.
\end{proof}

\begin{remark}
The upper bounds on the contraction coefficients for the parametric
$f$-divergences in \eqref{K div.} and \eqref{S div.} generalize the
upper bound on the contraction coefficient for the relative entropy
in \cite[Theorem~III.6]{Raginsky16} (recall that
$K_1(P\|Q) = D(P\|Q) = S_1(P\|Q)$), so the upper bounds in
Theorem~\ref{thm: contraction coefficients} are specialized to the
latter bound at $\alpha=1$.
\end{remark}

\begin{corollary} \label{corollary: bounds - contraction coefficient}
Let
\begin{align} \label{eq: 20200417a0}
\mu_{\chi^2}(W_{Y|X}) := \sup_Q \mu_{\chi^2}(Q_X, W_{Y|X}),
\end{align}
where the supremum in the right side is over all probability measures $Q_X$ defined on $\set{X}$.
Then,
\begin{align}
\label{eq: 20200417a1}
\mu_{\chi^2}(Q_X, W_{Y|X}) & \leq \mu_{k_\alpha}(Q_X, W_{Y|X}) \leq \mu_{\chi^2}(W_{Y|X}), \quad \alpha \in (0,1], \\
\label{eq: 20200417a2}
\mu_{\chi^2}(Q_X, W_{Y|X}) & \leq \mu_{s_\alpha}(Q_X, W_{Y|X}) \leq \mu_{\chi^2}(W_{Y|X}), \quad \; \alpha \in [0,1].
\end{align}
\end{corollary}
\begin{proof}
See Section~\ref{proof of Corollary: bounds - contraction coefficient}.
\end{proof}

\begin{example} \label{example: BSC}
Let $Q_X = \mathrm{Bernoulli}\bigl(\tfrac12\bigr)$,
and let $W_{Y|X}$ correspond to a binary symmetric channel (BSC) with
crossover probability $\varepsilon$. Then,
$\mu_{\chi^2}(Q_X, W_{Y|X}) = \mu_{\chi^2}(W_{Y|X}) = (1-2 \varepsilon)^2$.
The upper and lower bounds on $\mu_{k_\alpha}(Q_X, W_{Y|X})$ and
$\mu_{s_\alpha}(Q_X, W_{Y|X})$ in \eqref{eq: 20200417a1}
and \eqref{eq: 20200417a2} match for all $\alpha$,
and they are all equal to $(1-2 \varepsilon)^2$.
\end{example}

\vspace*{0.1cm}
The upper bound on the contraction coefficients in
Corollary~\ref{corollary: bounds - contraction coefficient} is given by
$\mu_{\chi^2}(W_{Y|X})$, whereas the lower bound is given by
$\mu_{\chi^2}(Q_X, W_{Y|X})$ which depends on the input distribution
$Q_X$. We next provide alternative upper bounds on the contraction
coefficients for the considered (parametric) $f$-divergences which,
similarly to the lower bound, scale like $\mu_{\chi^2}(Q_X, W_{Y|X})$.
Although the upper bound in Corollary~\ref{corollary: bounds - contraction coefficient}
may be tighter in some cases than the alternative upper bounds which are
next presented in Proposition~\ref{proposition: bounds - contraction coefficient}
(and in fact, the former upper bound may be even achieved with equality
as in Example~\ref{example: BSC}), the bounds in
Proposition~\ref{proposition: bounds - contraction coefficient}
are used shortly to determine the exponential rate of the convergence
to stationarity of a type of Markov chains.

\begin{proposition} \label{proposition: bounds - contraction coefficient}
For all $\alpha \in (0,1]$,
\begin{align}
\label{eq: 20200417b1}
& \mu_{\chi^2}(Q_X, W_{Y|X}) \leq \mu_{k_\alpha}(Q_X, W_{Y|X}) \leq
\frac1{\alpha \, Q_{\min}} \cdot \mu_{\chi^2}(Q_X, W_{Y|X}), \\
\label{eq: 20200417b2}
& \mu_{\chi^2}(Q_X, W_{Y|X}) \leq \mu_{s_\alpha}(Q_X, W_{Y|X}) \leq
\frac{(1-\alpha) \, \log_{\mathrm{e}}\Bigl(\frac1\alpha\Bigr)
+ 2\alpha-1}{(1-3 \alpha + 3 \alpha^2) \, Q_{\min}} \cdot \mu_{\chi^2}(Q_X, W_{Y|X}),
\end{align}
where $Q_{\min}$ denotes the minimal positive mass of the input distribution $Q_X$.
\end{proposition}
\begin{proof}
See Section~\ref{proof of proposition: bounds - contraction coefficient}.
\end{proof}

\begin{remark}
In view of \eqref{eq: 20200417c1}, at $\alpha=1$, \eqref{eq: 20200417b1} and
\eqref{eq: 20200417b2} specialize to an upper bound on the contraction coefficient
of the relative entropy (KL divergence) as a function of the contraction coefficient
of the chi-squared divergence. In this special case, both \eqref{eq: 20200417b1}
and \eqref{eq: 20200417b2} give
\begin{align}
\mu_{\chi^2}(Q_X, W_{Y|X}) \leq \mu_{\mathrm{KL}}(Q_X, W_{Y|X}) \leq \frac1{Q_{\min}}
\cdot \mu_{\chi^2}(Q_X, W_{Y|X}),
\end{align}
which then coincides with \cite[Theorem~10]{MakurZ15}.
\end{remark}

\vspace*{0.1cm}
We next apply Proposition~\ref{proposition: bounds - contraction coefficient} to consider
the convergence rate to stationarity of Markov chains by the introduced $f$-divergences
in Definition~\ref{def: parametric f-divergences}. The next result follows
\cite[Section~2.4.3]{Makur_PhD19}, and it provides a generalization of the
result there.

\begin{theorem} \label{thm: MC}
Consider a time-homogeneous, irreducible and reversible discrete-time Markov
chain with a finite state space $\set{X}$, let $W$ be its probability transition
matrix, and $Q_X$ be its unique stationary distribution (reversibility means that
$Q_X(x) [W]_{x,y} = Q_X(y) [W]_{y,x}$ for all $x,y \in \set{X}$). Let $P_X$ be
an initial probability distribution over $\set{X}$. Then, for all $\alpha \in (0,1]$
and $n \in \naturals$,
\begin{align}
\label{eq: 20200418a1}
& K_\alpha(P_X W^n \| Q_X) \leq \mu_{k_\alpha}(Q_X, W^n) \; K_\alpha(P_X \| Q_X), \\
\label{eq: 20200418a2}
& S_\alpha(P_X W^n \| Q_X) \leq \mu_{s_\alpha}(Q_X, W^n) \; S_\alpha(P_X \| Q_X),
\end{align}
and the contraction coefficients in the right sides of \eqref{eq: 20200418a1}
and \eqref{eq: 20200418a2} scale like the $n$-th power of the contraction coefficient
for the chi-squared divergence as follows:
\begin{align}
\label{eq: 20200418a3}
& \bigl( \mu_{\chi^2}(Q_X, W) \bigr)^n \leq \mu_{k_\alpha}(Q_X, W^n)
\leq \frac1{\alpha \, Q_{\min}} \cdot \bigl(\mu_{\chi^2}(Q_X, W)\bigr)^n, \\
\label{eq: 20200418a4}
& \bigl( \mu_{\chi^2}(Q_X, W) \bigr)^n \leq \mu_{s_\alpha}(Q_X, W^n)
\leq \frac{(1-\alpha) \, \log_{\mathrm{e}}\Bigl(\frac1\alpha\Bigr)
+ 2\alpha-1}{(1-3 \alpha + 3 \alpha^2) \, Q_{\min}} \cdot \bigl(\mu_{\chi^2}(Q_X, W)\bigr)^n.
\end{align}
\end{theorem}

\begin{proof}
Inequalities~\eqref{eq: 20200418a1} and \eqref{eq: 20200418a2} holds since $Q_X W^n = Q_X$,
for all $n \in \naturals$, and due to Definition~\ref{def: contraction} and
\eqref{eq: 20200412a1}--\eqref{eq: 20200412a2}.
Inequalities~\eqref{eq: 20200418a3} and \eqref{eq: 20200418a4} hold
by Proposition~\ref{proposition: bounds - contraction coefficient}, and due to the reversibility
of the Markov chain which implies that (see \cite[Eq.~(2.92)]{Makur_PhD19})
\begin{align}
\mu_{\chi^2}(Q_X, W^n) = \bigl( \mu_{\chi^2}(Q_X, W) \bigr)^n, \quad n \in \naturals.
\end{align}
\end{proof}

In view of \eqref{eq: 20200418a3} and \eqref{eq: 20200418a4}, Theorem~\ref{thm: MC}
readily gives the following result on the exponential decay rate of the upper bounds
on the divergences in the left sides of \eqref{eq: 20200418a1} and \eqref{eq: 20200418a2}.
\begin{corollary} \label{corollary: lim. of n-th root}
For all $\alpha \in (0,1]$,
\begin{align}
\lim_{n \to \infty} \bigl(\mu_{k_\alpha}(Q_X, W^n)\bigr)^{1/n}
= \mu_{\chi^2}(Q_X, W) = \lim_{n \to \infty} \bigl(\mu_{s_\alpha}(Q_X, W^n)\bigr)^{1/n}.
\end{align}
\end{corollary}

\begin{remark}
Theorem~\ref{thm: MC} and Corollary~\ref{corollary: lim. of n-th root} generalize
the results in \cite[Section~2.4.3]{Makur_PhD19}, which follow as a special case at $\alpha=1$
(see \eqref{eq: 20200417c1}).
\end{remark}

\vspace*{0.1cm}
We end this subsection by considering maximal correlations, which are closely related to
the contraction coefficient for the chi-squared divergence.
\begin{definition}  \label{definition: maximal correlation}
The {\em maximal correlation} between two random variables $X$ and $Y$
is defined as
\begin{align}
\label{eq: max. correlation}
\rho_{\mathrm{m}}(X;Y) := \sup_{f,g} \, \expectation[f(X) g(Y)],
\end{align}
where the supremum is taken over all real-valued functions $f$ and $g$ such that
\begin{align}
\expectation[f(X)] = \expectation[g(Y)] = 0,
\quad \expectation[f^2(X)] \leq 1, \; \expectation[g^2(Y)] \leq 1.
\end{align}
\end{definition}
It is well-known \cite{Sarmanov62} that if $X \sim Q_X$ and $Y \sim Q_Y = Q_X W_{Y|X}$,
then the contraction coefficient for the chi-squared divergence $\mu_{\chi^2}(Q_X, W_{Y|X})$
is equal to the square of the maximal correlation between the random variables $X$ and $Y$,
i.e.,
\begin{align} \label{eq:Sarmanov}
\rho_{\mathrm{m}}(X;Y) = \sqrt{\mu_{\chi^2}(Q_X, W_{Y|X})}.
\end{align}

A simple application of Corollary~\ref{corollary: RE and chi2} and \eqref{eq:Sarmanov}
gives the following result.
\begin{proposition} \label{prop: max. correlation}
In the setting of Definition~\ref{def: contraction}, for $s \in [0,1]$, let
$X_s \sim (1-s) P_X + s Q_X$ and $Y_s \sim (1-s) P_Y + s Q_Y$ with $P_X \neq Q_X$
and $P_X \ll \gg Q_X$. Then, the following inequality holds
\begin{align} \label{eq: LB max. correlation}
\sup_{s \in [0,1]} \rho_{\mathrm{m}}(X_s;Y_s) \geq \max \biggl\{
\sqrt{\frac{D(P_Y \| Q_Y)}{D(P_X \| Q_X)}}, \,
\sqrt{\frac{D(Q_Y \| P_Y)}{D(Q_X \| P_X)}} \biggr\}.
\end{align}
\end{proposition}
\begin{proof}
See Section~\ref{proof of prop: max. correlation}.
\end{proof}

\section{Proofs}
\label{section: proofs}

This section provides proofs of the results in Sections~\ref{section: main results}
and~\ref{section: applications}.

\subsection{Proof of Theorem~\ref{thm: KL and chi^2}}
\label{subsection: proof KL-chi2}

{\em Proof of \eqref{identity: KL}}: We rely on an integral representation
of the logarithm function (on base $\mathrm{e}$):
\begin{align} \label{int. rep. of log}
\log_{\mathrm{e}} x
&= \int_0^1 \frac{x-1}{x +(1-x)v} \; \mathrm{d}v, \quad \forall \, x >0.
\end{align}
Let $\mu$ be a dominating measure of $P$ and $Q$ (i.e., $P,Q \ll \mu$), and let
$p := \frac{\mathrm{d}P}{\mathrm{d}\mu}$,
$q := \frac{\mathrm{d}Q}{\mathrm{d}\mu}$, and
\begin{align}
\label{r_lambda}
r_\lambda := \frac{\mathrm{d}R_\lambda}{\mathrm{d}\mu}
= (1-\lambda)p + \lambda q, \quad \forall \, \lambda \in [0,1],
\end{align}
where the last equality is due to \eqref{mixture PMF}. For all
$\lambda \in [0,1]$,
\begin{align}
\label{1411a1}
\tfrac1{\log \mathrm{e}} \, D(P\| R_\lambda) &=
\int p \log_{\mathrm{e}} \Bigl( \frac{p}{r_\lambda}\Bigr) \, \mathrm{d}\mu \\
\label{1411a3}
&= \int_0^1 \int \frac{p (p - r_\lambda)}{p + v(r_\lambda - p)} \;
\mathrm{d}\mu \, \mathrm{d}v,
\end{align}
where \eqref{1411a3} holds due to \eqref{int. rep. of log} with
$x:=\frac{p}{r_\lambda}$, and by swapping the
order of integration. The inner integral in the right side of
\eqref{1411a3} satisfies, for all $v \in (0,1]$,
\begin{align}
& \int \frac{p (p - r_\lambda)}{p + v(r_\lambda - p)} \;
\mathrm{d}\mu \nonumber \\
&= \int (p-r_\lambda) \left( 1 + \frac{v(p-r_\lambda)}{p
+ v(r_\lambda - p)} \right)
\mathrm{d}\mu \\
&= \int (p-r_\lambda) \, \mathrm{d}\mu
+ v \int \frac{(p-r_\lambda)^2}{p + v(r_\lambda - p)} \; \mathrm{d}\mu \\
\label{1411a6}
&= v \int \frac{(p-r_\lambda)^2}{(1-v)p + v r_\lambda} \; \mathrm{d}\mu \\
&= \frac1{v} \int \frac{ \bigl(p-\bigl[(1-v)p
+ v r_\lambda \bigr] \bigr)^2}{(1-v)p
+ v r_\lambda} \; \mathrm{d}\mu \\[0.1cm]
\label{inner int. 1}
&= \frac1{v} \; \chi^2 \bigl( P \, \| \, (1-v)P + v R_\lambda \bigr),
\end{align}
where \eqref{1411a6} holds since $\int p \, \mathrm{d}\mu = 1$, and
$\int r_\lambda \, \mathrm{d}\mu = 1$. From \eqref{mixture PMF}, for
all $(\lambda , v) \in [0,1]^2$,
\begin{align}
\label{identity R}
(1-v)P + v R_\lambda
= (1-\lambda v) P + \lambda v \, Q = R_{\lambda v}.
\end{align}
The substitution of \eqref{identity R} into the right
side of \eqref{inner int. 1} gives that, for all
$(\lambda, v) \in [0,1] \times (0,1]$,
\begin{align}
\label{inner int. 2}
\int \frac{p (p - r_\lambda)}{p + v(r_\lambda - p)} \; \mathrm{d}\mu
= \frac1{v} \; \chi^2( P \| R_{\lambda v}).
\end{align}
Finally, substituting \eqref{inner int. 2} into the right side of
\eqref{1411a3} gives that, for all $\lambda \in (0,1]$,
\begin{align}
\label{1811b1}
\tfrac1{\log \mathrm{e}} \, D(P\| R_\lambda)
&= \int_0^1  \frac1{v} \; \chi^2( P \| R_{\lambda v}) \, \mathrm{d}v \\[0.1cm]
\label{1811b2}
&= \int_0^\lambda \frac1{s} \; \chi^2( P \| R_s) \, \mathrm{d}s,
\end{align}
where \eqref{1811b2} holds by the transformation $s := \lambda v$.
Equality~\eqref{1811b2} also holds for $\lambda = 0$ since
$D(P\|R_0) = D(P\|P) = 0$.

{\em Proof of \eqref{identity: chi2}}:
For all $s \in (0,1]$,
\begin{align}
\chi^2(P\|Q) &= \int \frac{(p-q)^2}{q} \; \mathrm{d}\mu \nonumber \\
&= \frac1{s^2} \int \frac{\bigl[ \bigl(sp + (1-s)q \bigr)
- q \bigr]^2}{q} \; \mathrm{d}\mu \\
\label{1411a4}
&= \frac1{s^2} \int \frac{ \bigl( r_{1-s} - q \bigr)^2}{q} \;
\mathrm{d}\mu \\[0.1cm]
\label{1411a5}
&= \frac1{s^2} \; \chi^2 \bigl( R_{1-s} \, \| \, Q \bigr),
\end{align}
where \eqref{1411a4} holds due to \eqref{r_lambda}. From
\eqref{1411a5}, it follows that for all $\lambda \in [0,1]$,
\begin{align}
\int_0^\lambda \frac1{s} \;
\chi^2 \bigl( R_{1-s} \, \| \, Q \bigr) \, \mathrm{d}s
= \int_0^\lambda s \, \mathrm{d}s \; \,  \chi^2(P\|Q)
= \tfrac12 \, \lambda^2 \, \chi^2(P\|Q).
\end{align}

\subsection{Proof of Proposition~\ref{prop: f-div. ineq.}}
\label{subsection: proof f-div ineq.}

\begin{enumerate}[1)]
\item
{\em Simple Proof of Pinsker's Inequality}:
By \cite{Gilardoni06-cor} or \cite[(58)]{ReidW11},
\begin{align}
\chi^2(P\|Q) \geq
\begin{dcases}
|P-Q|^2,  & \mbox{if} \; |P-Q| \in [0,1], \\[0.1cm]
\frac{|P-Q|}{2-|P-Q|}, & \mbox{if} \; |P-Q| \in (1, 2].
\end{dcases}
\end{align}
We need the weaker inequality $\chi^2(P\|Q) \geq |P-Q|^2$, proved by
the Cauchy-Schwarz inequality:
\begin{align}
\label{1411b1}
\chi^2(P\|Q) &= \int \frac{(p-q)^2}{q} \, \mathrm{d}\mu \; \int q \, \mathrm{d}\mu \\
\label{1411b2}
&\geq \left( \int \frac{|p-q|}{\sqrt{q}} \cdot \sqrt{q} \, \mathrm{d}\mu \right)^2 \\
\label{1411b4}
&= |P-Q|^2.
\end{align}
By combining \eqref{identity2: KL} and \eqref{1411b1}--\eqref{1411b4},
it follows that
\begin{align}
\tfrac1{\log \mathrm{e}} \, D(P\| Q)
&= \int_0^1 \chi^2(P \, \| \, (1-s)P + sQ) \;
\frac{\mathrm{d}s}{s} \\[0.1cm]
&\geq \int_0^1 \big|P - \bigl( (1-s)P + s Q \bigr) \big|^2 \;
\frac{\mathrm{d}s}{s} \\[0.1cm]
&= \int_0^1 s \, |P-Q|^2 \, \mathrm{d}s \\[0.1cm]
&= \tfrac12 \, |P-Q|^2.
\end{align}

\item
{\em Proof of \eqref{RD2-chi2 new} and its local tightness}:
\begin{align}
\label{1511a1}
\tfrac1{\log \mathrm{e}} \, D(P\| Q)
&= \int_0^1 \chi^2(P \, \| \, (1-s)P + sQ) \; \frac{\mathrm{d}s}{s} \\[0.1cm]
\label{1511a2}
&= \int_0^1 \left( \int \frac{\bigl[p-((1-s)p+sq)\bigr]^2}{(1-s)p + sq}
\; \mathrm{d}\mu \right) \frac{\mathrm{d}s}{s} \\[0.1cm]
\label{1511a3}
&= \int_0^1 \int \frac{s(p-q)^2}{(1-s)p + sq} \; \mathrm{d}\mu \;
\mathrm{d}s \\[0.1cm]
\label{1511a4}
&\leq \int_0^1 \int s(p-q)^2 \left( \frac{1-s}{p} + \frac{s}{q} \right)
\mathrm{d}\mu \; \mathrm{d}s \\[0.1cm]
\label{1511a5}
&= \int_0^1 s^2 \, \mathrm{d}s \; \int \frac{(p-q)^2}{q} \; \mathrm{d}\mu +
\int_0^1 s(1-s) \, \mathrm{d}s \; \int \frac{(p-q)^2}{p} \; \mathrm{d}\mu \\[0.1cm]
&= \tfrac13 \, \chi^2(P\|Q) + \tfrac16 \, \chi^2(Q\|P),
\end{align}
where \eqref{1511a1} is \eqref{identity2: KL}, and \eqref{1511a4} holds due to
Jensen's inequality and the convexity of the hyperbola.

We next show the local tightness of inequality \eqref{RD2-chi2 new} by proving
that \eqref{lim 1} yields \eqref{locally tight}. Let $\{P_n\}$ be a sequence of
probability measures, defined on a measurable space $(\set{X}, \mathscr{F})$,
and assume that $\{P_n\}$ converges to a probability measure $P$ in the sense
that \eqref{lim 1} holds. In view of \cite[Theorem~7]{Sason18} (see also
\cite{PardoV03} and \cite[Section~4.F]{ISSV16}), it follows that
\begin{align} \label{1811a1}
\lim_{n \to \infty} D(P_n \| P) = \lim_{n \to \infty} \chi^2(P_n \| P) = 0,
\end{align}
and
\begin{align}
\label{1811a2}
& \lim_{n \to \infty} \frac{D(P_n \| P)}{\chi^2(P_n \| P)}
= \tfrac12 \, \log \mathrm{e}, \\[0.1cm]
\label{1811a3}
& \lim_{n \to \infty} \frac{\chi^2(P_n \| P)}{\chi^2(P \| P_n)} = 1,
\end{align}
which therefore yields \eqref{locally tight}.

\item
{\em Proof of \eqref{LB-RE-chi2} and \eqref{29032020b1}}:
The proof of \eqref{LB-RE-chi2} relies on \eqref{identity4: KL} and the following lemma.
\begin{lemma}
For all $s, \theta \in (0,1)$,
\begin{align} \label{eq: ratio GV01}
\frac{D_{\phi_s}(P\|Q)}{D_{\phi_\theta}(P\|Q)}
\geq \min \biggl\{ \frac{1-\theta}{1-s}, \frac{\theta}{s} \biggr\}.
\end{align}
\end{lemma}
\begin{proof}
\begin{align}
D_{\phi_s}(P\|Q) &= \int \frac{(p-q)^2}{(1-s)p+sq} \; \mathrm{d}\mu \\
&= \int \frac{(p-q)^2}{(1-\theta)p+\theta q} \;
\frac{(1-\theta)p+\theta q}{(1-s)p+sq} \; \mathrm{d}\mu \\
&\geq \min \biggl\{ \frac{1-\theta}{1-s},
\frac{\theta}{s} \biggr\} \, \int \frac{(p-q)^2}{(1-\theta)p+\theta q} \; \; \mathrm{d}\mu \\
&= \min \biggl\{ \frac{1-\theta}{1-s}, \frac{\theta}{s} \biggr\} \; D_{\phi_\theta}(P\|Q).
\end{align}
\end{proof}

From \eqref{identity4: KL} and \eqref{eq: ratio GV01}, for all $\theta \in (0,1)$,
\begin{align}
\tfrac1{\log \mathrm{e}} \, D(P\| Q)
&= \int_0^\theta s D_{\phi_s}(P\|Q) \, \mathrm{d}s
+ \int_\theta^1  s D_{\phi_s}(P\|Q) \, \mathrm{d}s \\
&\geq \int_0^\theta \frac{s \, (1-\theta)}{1-s} \cdot D_{\phi_\theta}(P\|Q) \, \mathrm{d}s
+ \int_\theta^1 \theta \, D_{\phi_\theta}(P\|Q) \, \mathrm{d}s \\
&= \biggl[ -\theta + \log_{\mathrm{e}} \biggl(\frac1{1-\theta}\biggr) \biggr] \, (1-\theta) \, D_{\phi_\theta}(P\|Q)
+ \theta (1-\theta) \, D_{\phi_\theta}(P\|Q) \\
&= (1-\theta) \, \log_{\mathrm{e}} \biggl(\frac1{1-\theta}\biggr) \, D_{\phi_\theta}(P\|Q).
\end{align}
This proves \eqref{LB-RE-chi2}. Furthermore, under the assumption in \eqref{lim 1},
for all $\theta \in [0,1]$,
\begin{align}
\lim_{n \to \infty}  \frac{D(P\|P_n)}{D_{\phi_\theta}(P\|P_n)}
& = \lim_{n \to \infty}  \frac{D(P\|P_n)}{\chi^2(P\|P_n)}
\; \lim_{n \to \infty} \frac{\chi^2(P\|P_n)}{D_{\phi_\theta}(P\|P_n)} \\
&= \tfrac12 \log \mathrm{e} \cdot \frac2{\phi_{\theta}''(1)} \label{29032020b2} \\
&= \tfrac12 \, \log \mathrm{e}, \label{29032020b3}
\end{align}
where \eqref{29032020b2} holds due to \eqref{1811a2} and the local behavior of $f$-divergences
\cite{PardoV03}, and \eqref{29032020b3} holds due to \eqref{eq: f - GV01} which implies that
$\phi_{\theta}''(1)=2$ for all $\theta \in [0,1]$. This proves \eqref{29032020b1}.

\item
{\em Proof of \eqref{eq: ISSV16}}: From \eqref{identity2: KL}, we get
\begin{align}
\label{1611a1}
& \tfrac1{\log \mathrm{e}} \, D(P\| Q) \nonumber \\
&= \int_0^1 \chi^2(P \, \| \, (1-s)P + sQ) \; \frac{\mathrm{d}s}{s} \\[0.1cm]
\label{1611a2}
&= \int_0^1 \bigl[ \chi^2(P \, \| \, (1-s)P + sQ) - s^2 \, \chi^2(P \| Q) \bigr]
\; \frac{\mathrm{d}s}{s} + \int_0^1 s \, \mathrm{d}s \; \chi^2(P\|Q) \\[0.1cm]
\label{1611a3}
&= \int_0^1 \bigl[ \chi^2(P \, \| \, (1-s)P + sQ) - s^2 \, \chi^2(P \| Q) \bigr]
\; \frac{\mathrm{d}s}{s} + \tfrac12 \, \chi^2(P\|Q).
\end{align}
Referring to the integrand of the first term in the right side of \eqref{1611a3},
for all $s \in (0,1]$,
\begin{align}
& \frac1s \, \bigl[ \chi^2(P \, \| \, (1-s)P + sQ) - s^2 \, \chi^2(P \| Q) \bigr]
\nonumber \\[0.1cm]
\label{1611a4}
&= s \int (p-q)^2 \biggl[ \frac1{(1-s)p+sq} - \frac1{q} \biggr] \, \mathrm{d}\mu \\[0.1cm]
\label{1611a5}
&= s(1-s) \int \frac{(q-p)^3}{q \bigl[(1-s)p+sq \bigr]} \; \mathrm{d}\mu \\[0.1cm]
\label{1611a6}
&= s(1-s) \int |q-p| \cdot \underbrace{\frac{|q-p|}{q}
\cdot \frac{q-p}{p+s(q-p)}}_{\leq \frac1s \, 1\{q \geq p\}} \; \mathrm{d}\mu \\
\label{1611a7}
&\leq (1-s) \int (q-p) \, 1\{q \geq p\} \, \mathrm{d}\mu \\
\label{1611a8}
&= \tfrac12 (1-s) \, |P-Q|,
\end{align}
where the last equality holds since the equality $\int (q-p) \, \mathrm{d}\mu = 0$ implies that
\begin{align}
\label{TV identities 1}
& \int (q-p) \, 1\{q \geq p\} \, \mathrm{d}\mu = \int (p-q) \, 1\{p \geq q\}
\, \mathrm{d}\mu \\
\label{TV identities 2}
&= \tfrac12 \int |p-q| \, \mathrm{d}\mu = \tfrac12 \, |P-Q|.
\end{align}
From \eqref{1611a4}--\eqref{1611a8}, an upper bound on the right side of
\eqref{1611a3} results in. This gives
\begin{align}
\tfrac1{\log \mathrm{e}} \, D(P\| Q)
& \leq \tfrac12 \int_0^1 (1-s) \, \mathrm{d}s \; |P-Q|
+ \tfrac12 \, \chi^2(P\|Q) \label{1611a9} \\[0.1cm]
&= \tfrac14 \, |P-Q| + \tfrac12 \, \chi^2(P\|Q).  \label{1611a10}
\end{align}

It should be noted that \cite[Theorem~2~a)]{ISSV16} shows that inequality
\eqref{eq: ISSV16} is tight. To that end, let $\varepsilon \in (0,1)$, and
define probability measures $P_\varepsilon$ and $Q_\varepsilon$ on the set
$\set{A} = \{0,1\}$ with $P_\varepsilon(1) = \varepsilon^2$ and
$Q_\varepsilon(1)=\varepsilon$. Then,
\begin{align} \label{eq: from ISSV16}
& \lim_{\varepsilon \downarrow 0} \frac{\tfrac1{\log \mathrm{e}}
\, D(P_\varepsilon \| Q_\varepsilon)}{\tfrac14 \, |P_\varepsilon-Q_\varepsilon|
+ \tfrac12 \, \chi^2(P_\varepsilon \| Q_\varepsilon)} = 1.
\end{align}
\end{enumerate}

\subsection{Proof of Theorem~\ref{theorem: KL LB}}
\label{subsection: proof KL-LB}

We first prove Item~a) in Theorem~\ref{theorem: KL LB}. In view of the
Hammersley--Chapman--Robbins lower bound on the $\chi^2$ divergence, for
all $\lambda \in [0,1]$
\begin{align}
\label{2011a6}
\chi^2\bigl(P \| (1-\lambda)P + \lambda Q \bigr) \geq \frac{\bigl( \expectation[X]
- \expectation[Z_\lambda] \bigr)^2}{\Var(Z_\lambda)},
\end{align}
where $X \sim P$, $Y \sim Q$ and $Z_\lambda \sim R_\lambda := (1-\lambda)P + \lambda Q$
is defined by
\begin{align}
\label{2011a4}
Z_\lambda :=
\begin{dcases}
X, & \quad \mbox{with probability } \hspace*{0.15cm} 1-\lambda, \\
Y, & \quad \mbox{with probability} \hspace*{0.15cm} \lambda.
\end{dcases}
\end{align}
For $\lambda \in [0,1]$,
\begin{align}
\label{2011b1}
\expectation[Z_\lambda]=(1-\lambda)m_P + \lambda m_Q,
\end{align}
and it can be verified that
\begin{align}
\label{2011b7}
\Var(Z_\lambda)=(1-\lambda) \sigma_P^2 + \lambda\sigma_Q^2+\lambda(1-\lambda) (m_P-m_Q)^2.
\end{align}
We now rely on identity \eqref{identity2: KL}
\begin{align}
\label{2011a3}
\tfrac1{\log \mathrm{e}} \, D(P\| Q)
&= \int_0^1 \chi^2(P \| (1-\lambda) P + \lambda Q) \; \frac{\mathrm{d}\lambda}{\lambda}
\end{align}
to get a lower bound on the relative entropy.
Combining \eqref{2011a6}, \eqref{2011b7} and \eqref{2011a3} yields
\begin{align}
\label{2011b9}
\tfrac1{\log \mathrm{e}} \, D(P\| Q)
\geq  (m_P-m_Q)^2\int_0^1 \frac{\lambda}{(1-\lambda) \sigma_P^2
+ \lambda\sigma_Q^2+\lambda(1-\lambda) (m_P-m_Q)^2} \, \mathrm{d}\lambda.
\end{align}
From \eqref{a} and \eqref{b}, we get
\begin{align}
\label{2011c3}
\int_0^1 \frac{\lambda}{(1-\lambda) \sigma_P^2 +
\lambda\sigma_Q^2+\lambda(1-\lambda) (m_P-m_Q)^2} \, \mathrm{d}\lambda
=\int_0^1\frac{\lambda}{(\alpha-a\lambda)(\beta + a\lambda)} \, \mathrm{d}\lambda,
\end{align}
where
\begin{align}
\label{2011c4}
\alpha&:=\sqrt{\sigma_P^2+\frac{b^2}{4a^2}}+\frac{b}{2a}, \\[0.1cm]
\label{2011c5}
\beta&:=\sqrt{\sigma_P^2+\frac{b^2}{4a^2}}-\frac{b}{2a}.
\end{align}
By using the partial fraction decomposition of the integrand in the right side
of \eqref{2011c3}, we get
\begin{align}
\label{2011d6}
D(P\| Q)
&\geq \frac{(m_P-m_Q)^2}{a^2}\biggl[\frac{\alpha}{\alpha + \beta}
\log \biggl(\frac{\alpha}{\alpha - a}\biggr) +\frac{\beta}{\alpha + \beta}
\log \biggl(\frac{\beta}{\beta + a}\biggr)\biggr] \\[0.1cm]
\label{2011d7}
&=\frac{\alpha}{\alpha + \beta} \log \biggl(\frac{\alpha}{\alpha - a}\biggr)
+ \frac{\beta}{\alpha + \beta}\log \biggl(\frac{\beta}{\beta + a}\biggr) \\[0.1cm]
\label{2011d8}
&= d\biggl(\frac{\alpha}{\alpha + \beta} \, \bigl\| \, \frac{\alpha-a}{\alpha + \beta}\biggr),
\end{align}
where \eqref{2011d6} holds by integration since $\alpha-a\lambda$ and
$\beta+a\lambda$ are both non-negative for all $\lambda \in [0,1]$.
To verify the latter claim, it should be noted that \eqref{a}
and the assumption that $m_P \neq m_Q$ imply that $a \neq 0$.
Since $\alpha, \beta > 0$, it follows that for all $\lambda \in [0,1]$,
either $\alpha-a\lambda > 0$ or $\beta+a\lambda > 0$ (if $a<0$, then
the former is positive, and if $a>0$, then the latter is positive).
By comparing the denominators of both integrands in the left and right
sides of \eqref{2011c3}, it follows that
$(\alpha-a\lambda)(\beta + a \lambda) \geq 0$ for all $\lambda \in [0,1]$.
Since the product of $\alpha-a\lambda$ and $\beta+a\lambda$ is non-negative
and at least one of these terms is positive, it follows that $\alpha-a\lambda$
and $\beta+a\lambda$ are both non-negative for all $\lambda \in [0,1]$.
Finally, \eqref{2011d7} follows from \eqref{a}.

If $m_P-m_Q\rightarrow 0$ and $\sigma_P\neq\sigma_Q$, then it follows from
\eqref{a} and \eqref{b} that $a\rightarrow 0$ and $b \to \sigma_P^2 - \sigma_Q^2 \neq 0$.
Hence, from \eqref{2011c4} and \eqref{2011c5},
$\alpha \geq \left|\frac{b}{a}\right| \rightarrow \infty$ and $\beta \rightarrow 0$,
which implies that the lower bound on $D(P\| Q)$ in \eqref{2011d8} tends to zero.

Letting $r:=\frac{\alpha}{\alpha+\beta}$ and $s:=\frac{\alpha-a}{\alpha+\beta}$, we obtain
that the lower bound on $D(P\|Q)$ in \eqref{LB-KL} holds. This bound is consistent
with the expressions of $r$ and $s$ in \eqref{r} and \eqref{s} since from
\eqref{v}, \eqref{2011c4} and \eqref{2011c5},
\begin{align}
\label{2011e1}
r&=\frac{\alpha}{\alpha+\beta}=\frac{v+\frac{b}{2a}}{2v} =\frac{1}{2}+\frac{b}{4av},
\end{align}
\begin{align}
\label{2011e2}
s=\frac{\alpha-a}{\alpha+\beta}=r-\frac{a}{\alpha+\beta}=r-\frac{a}{2v}.
\end{align}
It should be noted that $r,s\in[0,1]$. First, from \eqref{2011c4} and
\eqref{2011c5}, $\alpha$ and $\beta$ are positive if $\sigma_P \neq 0$,
which yields $r=\frac{\alpha}{\alpha+\beta} \in (0,1)$.
We next show that $s \in [0,1]$. Recall that $\alpha-a\lambda$ and
$\beta+a\lambda$ are both non-negative for all $\lambda \in [0,1]$.
Setting $\lambda=1$ yields $\alpha \geq a$, which (from \eqref{2011e2})
implies that $s \geq 0$. Furthermore, from \eqref{2011e2} and the
positivity of $\alpha+\beta$, it follows that $s \leq 1$ if and only
if $\beta \geq -a$. The latter holds since $\beta + a \lambda \geq 0$
for all $\lambda \in [0,1]$ (in particular, for $\lambda = 1$).
If $\sigma_P=0$, then it follows from \eqref{r}--\eqref{v}
that $v = \frac{b}{2|a|}$, $b = a^2 + \sigma_Q^2$, and (recall that $a \neq 0$)
\begin{enumerate}[1)]
\item If $a>0$, then $v = \frac{b}{2a}$ implies that
$r=\frac12 + \frac{b}{4av} = 1$, and $s = r-\frac{a}{2v} = 1-\frac{a^2}{b}
= \frac{\sigma_Q^2}{\sigma_Q^2 + a^2} \in [0,1]$;
\item if $a<0$, then $v = -\frac{b}{2a}$ implies that $r=0$, and
$s = r-\frac{a}{2v} = \frac{a^2}{b} = \frac{a^2}{a^2+\sigma_Q^2} \in [0,1]$.
\end{enumerate}

We next prove Item~b) in Theorem~\ref{theorem: KL LB} (i.e., the
achievability of the lower bound in \eqref{LB-KL}). To that end,
we provide a technical lemma, which can be verified by the reader.
\begin{lemma}
\label{lemma: KL LB}
Let $r,s$ be given in \eqref{r}--\eqref{v}, and let $u_{1,2}$ be given in \eqref{u_1,2}.
Then,
\begin{align}
\label{2011_lem2}
& (s-r)(u_1-u_2)=m_Q-m_P, \\[0.1cm]
\label{2011_lem3}
& u_1+u_2=m_P+m_Q+\frac{\sigma_Q^2-\sigma_P^2}{m_Q-m_P}.
\end{align}
\end{lemma}

\vspace*{0.2cm}
Let $X\sim P$ and $Y\sim Q$ be defined on a set
$\set{U}=\{u_1,u_2\}$ (for the moment, the values of $u_1$
and $u_2$ are not yet specified) with $P[X=u_1]=r$,
$P[X=u_2]=1-r$, $Q[Y=u_1]=s$, and $Q[Y=u_2]=1-s$. We
now calculate $u_1$ and $u_2$ such that
$\expectation[X]=m_P$ and $\Var(X)=\sigma_P^2$.
This is equivalent to
\begin{align}
\label{2011f2}
& ru_1 +(1-r)u_2=m_P, \\
\label{2011f3}
& ru_1^2 +(1-r)u_2^2=m_P^2+\sigma_P^2.
\end{align}
Substituting \eqref{2011f2} into the right side of \eqref{2011f3} gives
\begin{align}
\label{2011f4}
ru_1^2 +(1-r)u_2^2=\bigl[ru_1 +(1-r)u_2\bigr]^2+\sigma_P^2,
\end{align}
which, by rearranging terms, also gives
\begin{align}
\label{2011f6}
u_1-u_2 = \pm \sqrt{\frac{\sigma_P^2}{r(1-r)}}.
\end{align}
Solving simultaneously \eqref{2011f2} and \eqref{2011f6} gives
\begin{align}
\label{2011g1}
& u_1=m_P \pm \sqrt{\frac{(1-r)\sigma_P^2}{r}}, \\
\label{2011g2}
& u_2=m_P \mp \sqrt{\frac{r\sigma_P^2}{1-r}}.
\end{align}
We next verify that by setting $u_{1,2}$ as in \eqref{u_1,2},
one also gets (as desired) that $\expectation[Y]=m_Q$ and $\Var(Y)=\sigma_Q^2$.
From Lemma ~\ref{lemma: KL LB}, and from \eqref{2011f2} and \eqref{2011f3}, we have
\begin{align}
\label{2011h1}
\expectation[Y]
&=su_1+(1-s)u_2 \\[0.1cm]
\label{2011h2}
&=\bigl(ru_1+(1-r)u_2\bigr)+(s-r)(u_1-u_2) \\[0.1cm]
\label{2011h3}
&=m_P+(s-r)(u_1-u_2)=m_Q, \\[0.2cm]
\label{2011h4}
\expectation[Y^2]
&=su_1^2+(1-s)u_2^2 \\[0.1cm]
\label{2011h5}
&=ru_1^2+(1-r)u_2^2+(s-r)(u_1^2-u_2^2) \\[0.1cm]
\label{2011h6}
&=\expectation[X^2]+(s-r)(u_1-u_2)(u_1+u_2) \\[0.1cm]
\label{2011h7}
&=m_P^2+\sigma_P^2+(m_Q-m_P)\biggl(m_P+m_Q+\frac{\sigma_Q^2-\sigma_P^2}{m_Q-m_P}\biggr) \\[0.1cm]
\label{2011h9}
&=m_Q^2+\sigma_Q^2.
\end{align}
By combining \eqref{2011h3} and \eqref{2011h9}, we obtain $\Var(Y)=\sigma_Q^2$.
Hence, the probability mass functions $P$ and $Q$ defined on $\set{U}=\{u_1,u_2\}$
(with $u_1$ and $u_2$ in \eqref{u_1,2}) such that
\begin{align}
\label{2011i1}
P(u_1)=1-P(u_2)=r, \quad
Q(u_1)=1-Q(u_2)=s
\end{align}
satisfy the equality constraints in \eqref{eq: 1st and 2nd moments}, while also
achieving the lower bound on $D(P\|Q)$ that is equal to $d(r\|s)$.
It can be also verified that the  second option where
\begin{align}
\label{2011i3}
u_1=m_P-\sqrt{\frac{(1-r)\sigma_P^2}{r}}, \quad u_2=m_P+\sqrt{\frac{r\sigma_P^2}{1-r}}
\end{align}
does {\em not} yield the satisfiability of the conditions $\expectation[Y]=m_Q$ and
$\Var(Y)=\sigma_Q^2$, so there is only a unique pair of probability measures $P$
and $Q$, defined on a two-element set that achieves the lower bound in \eqref{LB-KL}
under the equality constraints in \eqref{eq: 1st and 2nd moments}.

We finally prove Item~c) in Theorem~\ref{theorem: KL LB}.
Let $m \in \Reals, \sigma_P^2$ and $\sigma_Q^2$ be selected arbitrarily such that $\sigma_Q^2 \geq \sigma_P^2$.
We construct probability measures $P_\varepsilon$ and $Q_\varepsilon$, depending on
a free parameter $\varepsilon$, with means
$m_P = m_Q := m$ and variances $\sigma_P^2$ and $\sigma_Q^2$, respectively (means
and variances are independent of $\varepsilon$),
and which are defined on a three-element set $\set{U} := \{u_1, u_2, u_3\}$ as follows:
\begin{align}
\label{20200502a2}
& P_\varepsilon(u_1) = r, \quad \, P_\varepsilon(u_2) = 1-r, \quad \quad \hspace*{0.3cm} P_\varepsilon(u_3) = 0, \\
\label{20200502a3}
& Q_\varepsilon(u_1) = s, \quad Q_\varepsilon(u_2) = 1-s-\varepsilon, \quad Q_\varepsilon(u_3) = \varepsilon,
\end{align}
with $\varepsilon > 0$. We aim to set the parameters $r, s, u_1, u_2$ and $u_3$ (as a function of
$m, \sigma_P, \sigma_Q$ and $\varepsilon$) such that
\begin{align} \label{20200502a4}
\lim_{\varepsilon \to 0^+} \, D(P_\varepsilon \| Q_\varepsilon) = 0.
\end{align}
Proving \eqref{20200502a4} yields \eqref{20200502a1}, while it also follows that the infimum in the left side of
\eqref{20200502a1} can be restricted to probability measures which are defined on a three-element set.

In view of the constraints on the means and variances in \eqref{eq: 1st and 2nd moments}, with
equal means $m$, we get the following set of equations from \eqref{20200502a2} and \eqref{20200502a3}:
\begin{align}
\label{20200502a5}
\begin{cases}
& r u_1 + (1-r) u_2 = m, \\
& s u_1 + (1-s-\varepsilon) u_2 + \varepsilon u_3 = m, \\
& r u_1^2 + (1-r) u_2^2 = m^2 + \sigma_P^2, \\
& s u_1^2 + (1-s-\varepsilon) u_2^2 + \varepsilon u_3^2 = m^2 + \sigma_Q^2.
\end{cases}
\end{align}
The first and second equations in \eqref{20200502a5} refer to the equal means
under $P$ and $Q$, and the third and fourth equations in \eqref{20200502a5} refer
to the second moments in \eqref{eq: 1st and 2nd moments}. Furthermore, in view
of \eqref{20200502a2} and \eqref{20200502a3}, the relative entropy is given by
\begin{align}
\label{20200502a6}
D(P_\varepsilon \| Q_\varepsilon) = r \log \frac{r}{s} + (1-r) \log \frac{1-r}{1-s-\varepsilon}.
\end{align}
Subtracting the square of the first equation in \eqref{20200502a5} from its third equation
gives the equivalent set of equations
\begin{align}
\label{20200502a7}
\begin{cases}
& r u_1 + (1-r) u_2 = m, \\
& s u_1 + (1-s-\varepsilon) u_2 + \varepsilon u_3 = m, \\
& r (1-r) (u_1 - u_2)^2 = \sigma_P^2, \\
& s u_1^2 + (1-s-\varepsilon) u_2^2 + \varepsilon u_3^2 = m^2 + \sigma_Q^2.
\end{cases}
\end{align}
We next select $u_1$ and $u_2$ such that $u_1 - u_2 := 2 \sigma_P$.
Then, the third equation in \eqref{20200502a7} gives $r(1-r) = \tfrac{1}{4}$, so
$r = \frac12$. Furthermore, the first equation in \eqref{20200502a7} gives
\begin{align}
\label{20200502a10}
& u_1 = m + \sigma_P, \\
\label{20200502a11}
& u_2 = m - \sigma_P.
\end{align}
Since $r$, $u_1$ and $u_2$ are independent of $\varepsilon$, so is the probability measure $P_\varepsilon := P$.
Combining the second equation in \eqref{20200502a7} with \eqref{20200502a10} and \eqref{20200502a11} gives
\begin{align}
\label{20200502a13}
u_3 = m - \biggl(1 + \frac{2s-1}{\varepsilon} \biggr) \sigma_P.
\end{align}
Substituting \eqref{20200502a10}--\eqref{20200502a13} into the fourth equation
of \eqref{20200502a7} gives a quadratic equation for $s$, whose selected solution
(such that $s$ and $r = \tfrac12$ be close for small $\epsilon>0$) is equal to
\begin{align}
\label{20200502a14}
s = \tfrac12 \left[1-\varepsilon + \sqrt{\left(\frac{\sigma_Q^2}{\sigma_P^2}-1+\varepsilon \right) \varepsilon} \, \right].
\end{align}
Hence, $s = \tfrac12 + O(\sqrt{\varepsilon})$, which implies that $s \in (0, 1-\varepsilon)$ for sufficiently small
$\varepsilon > 0$ (as it is required in \eqref{20200502a3}). In view of \eqref{20200502a6}, it also follows that
$D(P \| Q_\varepsilon)$ vanishes as we let $\varepsilon$ tend to zero.

We finally outline an alternative proof, which refers
to the case of equal means with arbitrarily selected $\sigma_P^2$ and $\sigma_Q^2$.
Let $(\sigma_P^2,\sigma_Q^2) \in (0,\infty)^2$. We next construct a sequence of pairs of
probability measures $\{(P_n,Q_n)\} $ with zero mean and respective variances
$(\sigma_P^2,\sigma_Q^2)$ for which $D(P_n \| Q_n ) \to 0$ as $n \to \infty$ (without any loss
of generality, one can assume that the equal means are equal to zero). We start by assuming
$(\sigma_P^2,\sigma_Q^2) \in (1,\infty)^2$. Let
\begin{align}
\mu_n :=  \sqrt{1 + n \bigl(\sigma_Q^2 - 1\bigr)},
\end{align}
and define a sequence of quaternary real-valued random variables with probability mass functions
\begin{align}\label{4Q_n}
Q_n (a) :=
\begin{dcases}
\tfrac12 - \tfrac{1}{2n} & a = \pm 1,\\
\tfrac1{2n} & a=  \pm \mu_n.
\end{dcases}
\end{align}
It can be verified that, for all $n \in \naturals$, $Q_n$ has zero mean and variance $\sigma_Q^2$. Furthermore, let
\begin{align}
P_n (a) :=
\begin{dcases}
\tfrac12 - \tfrac{\xi}{2n} & a = \pm 1,\\
\tfrac{\xi}{2n} & a=  \pm \mu_n,
\end{dcases}
\end{align}
with
\begin{align}
\xi := \frac{\sigma_P^2-1}{\sigma_Q^2-1} .
\end{align}
If $\xi >1$, for $n =1 , \ldots, \lceil\xi \rceil$, we choose $P_n$ arbitrarily
with mean 0 and variance $\sigma_P^2$. Then,
\begin{align}
& \Var(P_n) = 1 - \tfrac{\xi}{n} + \tfrac{\xi}{n} \mu^2_n = \sigma_P^2, \\
& D(P_n\| Q_n) = d \left( \frac{\xi}{n} \bigg\| \frac{1}{n} \right) \to 0.
\end{align}
Next suppose
$\min \{ \sigma_P^2,\sigma_Q^2 \} := \sigma^2 < 1$,
then
construct $P^\prime_n$ and $Q^\prime_n$ as
before with variances
$\frac{2 \sigma_P^2}{\sigma^2} > 1$ and
$\frac{2 \sigma_Q^2}{\sigma^2}>1$, respectively.
If $P_n$ and $Q_n$ denote
the random variables $P^\prime_n$ and $Q^\prime_n$
scaled by a factor of $\frac{\sigma}{\sqrt{2}}$, then
their variances are $\sigma_P^2$, $\sigma_Q^2$,
respectively, and $D(P_n\|Q_n) = D(P^\prime_n\|Q^\prime_n) \to 0$
as we let $n \to \infty$.

To conclude, it should be noted that the sequences of probability measures in the
latter proof are defined on a four-element set. Recall that in the earlier proof,
specialized to the case of (equal means with) $\sigma_P^2 \leq \sigma_Q^2$, the introduced
probability measures are defined on a three-element set, and the reference probability
measure $P$ is fixed while referring to an equiprobable binary random variable.

\subsection{Proof of Theorem~\ref{thm: F}}
\label{subsection: proof of thm: F}

We first prove \eqref{diff F 2a}. Differentiating
both sides of \eqref{identity: KL} gives that, for all $\lambda \in (0,1]$,
\begin{align}
\label{1711a1}
F'(\lambda) &= \frac1{\lambda} \; \chi^2\bigl(P \| R_\lambda \bigr)
\, \log \mathrm{e} \\[0.1cm]
\label{1711a2}
&\geq \frac1{\lambda} \Bigl[ \exp\bigl( D(P \| R_\lambda) \bigr) - 1 \Bigr]
\, \log \mathrm{e} \\[0.1cm]
\label{1711a3}
&= \frac1{\lambda} \Bigl[ \exp\bigl(F(\lambda)\bigr) - 1 \Bigr]
\, \log \mathrm{e},
\end{align}
where \eqref{1711a1} holds due to \eqref{mixture PMF}, \eqref{identity: KL}
and \eqref{def:F}; \eqref{1711a2} holds by \eqref{RD2-chi2}, and \eqref{1711a3}
is due to \eqref{mixture PMF} and \eqref{def:F}. This gives \eqref{diff F 2a}.

We next prove \eqref{diff F 2b}, and the conclusion which appears after it.
In view of \cite[Theorem~8]{Sason18}, applied to $f(t) := -\log t$ for all $t>0$,
we get (it should be noted that, by the definition of $F$ in \eqref{def:F}, the result
in \cite[(195)--(196)]{Sason18} is used here by swapping $P$ and $Q$)
\begin{align}
\label{20200421a1}
\lim_{\lambda \to 0^+} \frac{F(\lambda)}{\lambda^2}
= \tfrac12 \, \chi^2(Q\|P) \, \log \mathrm{e}.
\end{align}
Since $\underset{\lambda \to 0^+}{\lim} F(\lambda) = 0$, it follows by
L'H\^{o}pital's rule that
\begin{align}
\lim_{\lambda \to 0^+} \frac{F'(\lambda)}{\lambda}
= 2 \lim_{\lambda \to 0^+} \frac{F(\lambda)}{\lambda^2}
= \chi^2(Q\|P) \, \log \mathrm{e},
\end{align}
which gives \eqref{diff F 2b}.
A comparison of the limit in \eqref{diff F 2b} with a lower bound which
follows from \eqref{diff F 2a} gives
\begin{align}
\lim_{\lambda \to 0^+} \frac{F'(\lambda)}{\lambda} &\geq
\lim_{\lambda \to 0^+} \frac1{\lambda^2} \Bigl[ \exp\bigl(F(\lambda)\bigr) - 1 \Bigr]
\, \log \mathrm{e} \\
&= \lim_{\lambda \to 0^+} \frac{F(\lambda)}{\lambda^2} \,
\lim_{\lambda \to 0^+} \frac{\exp\bigl(F(\lambda)\bigr) - 1}{F(\lambda)} \cdot \log \mathrm{e} \\
&= \lim_{\lambda \to 0^+} \frac{F(\lambda)}{\lambda^2} \, \lim_{u \to 0} \frac{\mathrm{e}^u-1}{u} \\
\label{LB limit}
&= \tfrac12 \, \chi^2(Q\|P) \, \log \mathrm{e},
\end{align}
where \eqref{LB limit} relies on \eqref{20200421a1}.
Hence, the limit in \eqref{diff F 2b} is twice larger than its lower bound in the
right side of \eqref{LB limit}. This proves the conclusion which comes right after \eqref{diff F 2b}.

We finally prove \eqref{UB on F} based on \eqref{diff F 2a}. The function
$F$ is non-negative on $[0,1]$, and it is strictly positive on $(0,1]$
if $P \neq Q$. Let $P \neq Q$ (otherwise, \eqref{UB on F} is trivial).
Rearranging terms in \eqref{diff F 2a} and integrating both sides over
the interval $[\lambda, 1]$, for $\lambda \in (0,1]$, gives that
\begin{align}
\label{1711b1}
\int_\lambda^1 \frac{F'(t)}{\exp\bigl(F(t)\bigr)-1} \; \mathrm{d}t
&\geq \int_\lambda^1 \frac{\mathrm{d}t}{t} \; \log \mathrm{e} \\
\label{1711b2}
&= \log \frac{1}{\lambda}, \quad \forall \, \lambda \in (0,1].
\end{align}
The left side of \eqref{1711b1} satisfies
\begin{align}
\label{1711b3}
\int_\lambda^1 \frac{F'(t)}{\exp\bigl(F(t)\bigr)-1} \; \mathrm{d}t
&= \int_\lambda^1 \frac{F'(t) \,
\exp\bigl(-F(t)\bigr)}{1-\exp\bigl(-F(t)\bigr)} \; \mathrm{d}t \\[0.1cm]
\label{1711b4}
&= \int_\lambda^1  \frac{\mathrm{d}}{\mathrm{d}t} \Bigl\{ \log
\Bigl(1-\exp\bigl(-F(t)\bigr) \Bigr) \Bigr\} \; \mathrm{d}t \\[0.1cm]
\label{1711b6}
&= \log \Biggl( \frac{1-\exp\bigl(-D(P\|Q)\bigr)}{1-
\exp\bigl(-F(\lambda)\bigr)} \Biggr),
\end{align}
where \eqref{1711b6} holds since $F(1) = D(P\|Q)$ (see \eqref{def:F}).
Combining \eqref{1711b1}--\eqref{1711b6} gives
\begin{align}
\label{1711b7}
\frac{1-\exp\bigl(-D(P\|Q)\bigr)}{1-\exp\bigl(-F(\lambda)\bigr)}
\geq \frac{1}{\lambda}, \quad \forall \, \lambda \in (0,1],
\end{align}
which, due to the non-negativity of $F$, gives the right side inequality
in \eqref{UB on F} after rearrangement of terms in \eqref{1711b7}.

\subsection{Proof of Theorem~\ref{theorem: monotonic_f_divergence_seq}}
\label{subsection: proof monotonic_f_divergence_seq}

\begin{lemma}
\label{lemma: integral_convex_function}
Let $f_0 \colon (0, \infty) \rightarrow \Reals$ be a convex function with $f_0(1)=0$,
and let $\{ f_k(\cdot) \}_{k=0}^\infty$ be defined as in \eqref{eq:recursion f_k}.
Then, $\{f_k(\cdot)\}_{k=0}^{\infty}$ is a sequence of convex functions on $(0, \infty)$, and
\begin{align} \label{eq:inequality_f}
f_k(x) \geq f_{k+1}(x), \quad \forall \, x > 0, \; \; k \in \{0, 1, \ldots\}.
\end{align}
\end{lemma}
\begin{proof}
We prove the convexity of $\{f_k(\cdot)\}$ on $(0, \infty)$ by induction.
Suppose that $f_k(\cdot)$ is a convex function with $f_k(1)=0$ for a fixed integer $k \geq 0$.
The recursion in \eqref{eq:recursion f_k} yields $f_{k+1}(1)=0$ and, by the change
of integration variable $s :=(1-x)s'$,
\begin{align}
\label{2012a1}
f_{k+1}(x) &= \int_0^1 f_k(s'x-s'+1) \, \frac{\mathrm{d}s'}{s'}, \quad x>0.
\end{align}
Consequently, for $t\in (0,1)$ and $x\neq y$ with $x,y>0$, applying \eqref{2012a1} gives
\begin{align}
f_{k+1}((1-t)x+ty)&=\int_0^{1} f_k\bigl(s' [(1-t)x+ty]-s'+1\bigr)
\, \frac{\mathrm{d}s'}{s'} \label{eq:20200409a1} \\[0.1cm]
&=\int_0^1 f_k\bigl((1-t)(s'x-s'+1)+t(s'y-s'+1)\bigr) \,
\frac{\mathrm{d}s'}{s'} \label{eq:20200409a2} \\[0.1cm]
&\leq (1-t)\int_0^1 f_k(s'x-s'+1) \, \frac{\mathrm{d}s'}{s'}
+ t\int_0^1 f_k(s'y-s'+1) \, \frac{\mathrm{d}s'}{s'} \label{eq:20200409a3} \\[0.1cm]
&=(1-t)f_{k+1}(x)+tf_{k+1}(y),  \label{eq:20200409a4}
\end{align}
where \eqref{eq:20200409a3} holds since $f_k(\cdot)$ is convex on $(0, \infty)$
(by assumption). Hence, from \eqref{eq:20200409a1}--\eqref{eq:20200409a4},
$f_{k+1}(\cdot)$ is also convex on $(0, \infty)$ with $f_{k+1}(1)=0$.
By induction and our assumptions on $f_0$, it follows that
$\{f_k(\cdot)\}_{k=0}^{\infty}$ is a sequence of convex functions on $(0, \infty)$
which vanish at~1.

We next prove \eqref{eq:inequality_f}. For all $x, y > 0$ and $k \in \{0, 1, \ldots\}$,
\begin{align}
\label{eq:20200409a5}
f_{k+1}(y) &\geq f_{k+1}(x)+f'_{k+1}(x) \, (y-x) \\[0.1cm]
\label{eq:20200409a6}
&=f_{k+1}(x)+\frac{f_{k}(x)}{x-1} \; (y-x),
\end{align}
where \eqref{eq:20200409a5} holds since $f_k(\cdot)$ is convex on $(0, \infty)$, and
\eqref{eq:20200409a6} relies on the recursive equation in \eqref{eq:recursion f_k}.
Substituting $y=1$ into \eqref{eq:20200409a5}--\eqref{eq:20200409a6} and using the
equality $f_{k+1}(1)=0$ gives \eqref{eq:inequality_f}.
\end{proof}

We next prove Theorem~\ref{theorem: monotonic_f_divergence_seq}.
From Lemma~\ref{lemma: integral_convex_function}, it follows that $D_{f_k}(P\|Q)$
is an $f$-divergence for all integers $k\geq 0$, and the non-negative sequence
$\bigl\{D_{f_k}(P\|Q)\}_{k=0}^{\infty}$ is monotonically non-increasing. From
\eqref{mixture PMF} and \eqref{eq:recursion f_k}, it also follows that for all
$\lambda \in [0,1]$ and integer $k \in \{0, 1, \ldots\}$,
\begin{align}
D_{f_{k+1}}(R_\lambda\| P) &= \int p \, f_{k+1} \Bigl(\frac{r_\lambda}{p}\Bigr)
\; \mathrm{d}\mu \label{eq:20200409a7} \\[0.1cm]
&= \int p\int_0^{(p-q)\lambda/p} f_{k} (1-s) \; \frac{\mathrm{d}s}{s}
\; \mathrm{d}\mu \label{eq:20200409a8} \\[0.1cm]
&= \int p\int_0^{\lambda} f_{k} \Bigl(1+\frac{(q-p)s'}{p}\Bigr)
\, \frac{\mathrm{d}s'}{s'} \; \mathrm{d}\mu \label{eq:20200409a9} \\[0.1cm]
&= \int_0^{\lambda} \int p f_{k} \Bigl(\frac{r_{s'}}{p}\Bigr) \; \mathrm{d}\mu
\; \frac{\mathrm{d}s'}{s'} \label{eq:20200409a10} \\[0.1cm]
&=\int_0^{\lambda} D_{f_k}(R_{s'}\|P) \; \frac{\mathrm{d}s'}{s'},
\end{align}
where the substitution $s:=\frac{(p-q)s'}{p}$ is invoked in \eqref{eq:20200409a9},
and then \eqref{eq:20200409a10} holds by the equality
$\frac{r_{s'}}{p}=1+\frac{(q-p) \, s'}{p}$ for $s' \in [0,1]$ (this follows from
\eqref{mixture PMF}) and by interchanging the order of the integrations.

\subsection{Proof of Corollary~\ref{corollary: polylogarithm}}
\label{proof: corollary-polylogarithm}

Combining \eqref{eq:polylog} and \eqref{eq: def f_k} yields \eqref{eq:recursion f_k};
furthermore, $f_0 \colon (0, \infty) \to \Reals$, given by $f_0(x) = \frac{1}{x}-1$
for all $x>0$, is convex on $(0, \infty)$ with $f_0(1)=0$. Hence,
Theorem~\ref{theorem: monotonic_f_divergence_seq} holds for the selected functions
$\{f_k(\cdot)\}_{k=0}^{\infty}$ in \eqref{eq: def f_k}, which therefore are all convex
on $(0, \infty)$ and vanish at~1. This proves that \eqref{eq:fD_integral} holds for
all $\lambda \in[0,1]$ and $k \in \{0, 1, \ldots\}$. Since $f_0(x) = \frac{1}{x}-1$
and $f_1(x) = -\log_{\mathrm{e}}(x)$ for all $x>0$ (see \eqref{eq:polylog} and
\eqref{eq: def f_k}), then for every pair of probability measures $P$ and $Q$:
\begin{align} \label{reverse KL and chi^2}
D_{f_0}(P\|Q) = \chi^2(Q\|P), \quad D_{f_1}(P\|Q) = \tfrac1{\log \mathrm{e}} \; D(Q\|P).
\end{align}
Finally, combining \eqref{eq:fD_integral}, for $k=0$, together with \eqref{reverse KL and chi^2}
gives \eqref{identity: KL} as a special case.

\subsection{Proof of Theorem~\ref{thm: gen. Csiszar} and Corollary~\ref{corollary: gen. Csiszar}}
\label{subsection: proof of Csiszar84 extended}

For an arbitrary measurable set $\set{E} \subseteq \set{X}$, we have from \eqref{cond. PM}
\begin{align}
\mu_{\set{C}}(\set{E}) = \int_{\set{E}} \frac{1_{\set{C}}(x)}{\mu(\set{C})} \; \mathrm{d}\mu(x),
\end{align}
where $1_{\set{C}} \colon \set{X} \to \{0,1\}$ is the indicator function of the set
$\set{C} \subseteq \set{X}$, i.e., $1_{\set{C}}(x) := 1\{x \in \set{C}\}$ for all $x \in \set{X}$.
Hence,
\begin{align}
\frac{\mathrm{d}\mu_{\set{C}}}{\mathrm{d}\mu} \, (x) = \frac{1_{\set{C}}(x)}{\mu(\set{C})},
\quad \forall \, x \in \set{X},
\end{align}
and
\begin{align}
D(\mu_{\set{C}} \| \mu)
&= \int_{\set{X}} f\Bigl(\frac{\mathrm{d}\mu_{\set{C}}}{\mathrm{d}\mu}\Bigr) \; \mathrm{d}\mu \\[0.1cm]
&= \int_{\set{C}} f\biggl(\frac1{\mu(\set{C})}\biggr) \; \mathrm{d}\mu(x)
+ \int_{\set{X} \setminus \set{C}} f(0) \; \mathrm{d}\mu(x) \\[0.1cm]
&= \mu(\set{C}) \; f\biggl(\frac1{\mu(\set{C})}\biggr)
+ \mu(\set{X} \setminus \set{C}) \; f(0) \\[0.1cm]
&= \widetilde{f}\bigl(\mu(\set{C})\bigr) + (1-\mu(\set{C})) \, f(0),
\end{align}
where the last equality holds by the definition of $\widetilde{f}$ in \eqref{f conjugate}.
This proves Theorem~\ref{thm: gen. Csiszar}. Corollary~\ref{corollary: gen. Csiszar} is next
proved by first proving \eqref{eq: Csiszar84c} for the R\'{e}nyi divergence.
For all $\alpha \in (0,1) \cup (1, \infty)$,
\begin{align}
D_{\alpha}\bigl( \mu_{\set{C}} \| \mu \bigr) &=
\frac1{\alpha-1} \; \log \int_{\set{X}} \biggl( \frac{\mathrm{d}\mu_{\set{C}}}{\mathrm{d}\mu}
\biggr)^\alpha \, \mathrm{d}\mu \\[0.1cm]
&= \frac1{\alpha-1} \; \log \int_{\set{C}} \biggl( \frac1{\mu(\set{C})} \biggr)^\alpha
\, \mathrm{d}\mu \\[0.1cm]
&= \frac1{\alpha-1} \; \log \biggl( \biggl( \frac1{\mu(\set{C})} \biggr)^\alpha
\, \mu(\set{C}) \biggr) \\[0.1cm]
&= \log \frac1{\mu(\set{C})}.
\end{align}
The justification of \eqref{eq: Csiszar84c} for $\alpha=1$ is due to the continuous
extension of the order-$\alpha$ R\'{e}nyi divergence at $\alpha=1$, which gives the
relative entropy (see \eqref{def:d1}). Equality \eqref{eq: Csiszar84a} is obtained
from \eqref{eq: Csiszar84c} at $\alpha=1$. Finally, \eqref{eq: Csiszar84b} is obtained
by combining \eqref{RD2} and \eqref{eq: Csiszar84c} with $\alpha = 2$.

\subsection{Proof of Theorem~\ref{thm: contraction coefficients}}
\label{subsection: proof of thm: contraction coefficients}

Eq.~\eqref{eq: 20200412c1} is an equivalent form of \eqref{identity3: KL}.
From \eqref{S-K div.} and \eqref{eq: 20200412c1}, for all
$\alpha \in [0,1]$
\begin{align}
\tfrac1{\log \mathrm{e}} \, S_\alpha(P\|Q)
&= \alpha \; \tfrac1{\log \mathrm{e}} \, K_\alpha(P\|Q)
+ (1-\alpha) \; \tfrac1{\log \mathrm{e}} \, K_{1-\alpha}(Q \| P) \\
&= \alpha \int_0^\alpha s D_{\phi_s}(P\|Q) \, \mathrm{d}s +
(1-\alpha) \int_0^{1-\alpha} s D_{\phi_s}(Q\|P) \, \mathrm{d}s \\
\label{eq: 20200412e1}
&= \alpha \int_0^\alpha s D_{\phi_s}(P\|Q) \, \mathrm{d}s +
(1-\alpha) \int_\alpha^1 (1-s) D_{\phi_{1-s}}(Q\|P) \, \mathrm{d}s.
\end{align}
Regarding the integrand of the second term in \eqref{eq: 20200412e1},
in view of \eqref{eq: GV01-chi^2}, for all $s \in (0,1)$
\begin{align}
D_{\phi_{1-s}}(Q\|P)
&= \frac1{(1-s)^2} \cdot \chi^2\bigl(Q \, \| \, (1-s)P+sQ\bigr) \\
\label{eq: identity chi^2}
&= \frac1{s^2} \cdot \chi^2\bigl(P \, \| \, (1-s)P+sQ\bigr) \\
\label{eq: 20200412e2}
&= D_{\phi_s}(P\|Q),
\end{align}
where \eqref{eq: identity chi^2} readily follows from \eqref{eq: chi-square 1b}.
Since we also have $D_{\phi_1}(P\|Q) = \chi^2(P\|Q) = D_{\phi_0}(Q\|P)$ (see
\eqref{eq: GV01-chi^2}), it follows that
\begin{align}
D_{\phi_{1-s}}(Q\|P) = D_{\phi_s}(P\|Q),  \quad s \in [0,1].
\end{align}
By using this identity, we get from \eqref{eq: 20200412e1} that for all
$\alpha \in [0,1]$
\begin{align}
\tfrac1{\log \mathrm{e}} \, S_\alpha(P\|Q)
&= \alpha \int_0^\alpha s D_{\phi_s}(P\|Q) \, \mathrm{d}s +
(1-\alpha) \int_\alpha^1 (1-s) D_{\phi_s}(P\|Q) \, \mathrm{d}s \\
&= \int_0^1 g_\alpha(s) \, D_{\phi_s}(P\|Q) \, \mathrm{d}s,
\end{align}
where the function $g_\alpha \colon [0,1] \to \Reals$ is defined in
\eqref{eq: 20200412c3}. This proves the integral identity \eqref{eq: 20200412c2}.

The lower bounds in \eqref{eq: 20200412d1} and \eqref{eq: 20200412d2}
hold since if $f \colon (0, \infty) \to \Reals$ is convex, continuously
twice differentiable and strictly convex at~1, then
\begin{align} \label{eq: CohenKZ98}
\mu_{\chi^2}(Q_X, W_{Y|X}) \leq \mu_f(Q_X, W_{Y|X}),
\end{align}
(see, e.g., \cite[Proposition~II.6.5]{CohenKZ98} and \cite[Theorem~2]{PolyanskiyW17}).
Hence, this holds in particular for the $f$-divergences in \eqref{eq: 20200412a1}
and \eqref{eq: 20200412a2} (since the required properties are satisfied
by the parametric functions in \eqref{eq: 20200412b1} and \eqref{eq: 20200412b2},
respectively). We next prove the upper bound on the contraction coefficients
in \eqref{eq: 20200412d1} and \eqref{eq: 20200412d2} by relying on
\eqref{eq: 20200412c1} and \eqref{eq: 20200412c2}, respectively.
In the setting of Definition~\ref{def: contraction}, if $P_X \neq Q_X$,
then it follows from \eqref{eq: 20200412c1} that for $\alpha \in (0,1]$,
\begin{align}
\label{eq: 202020413a1}
\frac{K_\alpha(P_Y \| Q_Y)}{K_\alpha(P_X \| Q_X)}
&=  \frac{\int_0^\alpha s D_{\phi_s}(P_Y\|Q_Y) \,
\mathrm{d}s}{\int_0^\alpha s D_{\phi_s}(P_X\|Q_X) \, \mathrm{d}s} \\
&\leq \frac{\int_0^\alpha s \, \mu_{\phi_s}(Q_X, W_{Y|X}) \, D_{\phi_s}(P_X\|Q_X) \,
\mathrm{d}s}{\int_0^\alpha s D_{\phi_s}(P_X\|Q_X) \, \mathrm{d}s} \\
&\leq \sup_{s \in (0, \alpha]} \mu_{\phi_s}(Q_X, W_{Y|X}).
\end{align}
Finally, supremizing the left-hand side of \eqref{eq: 202020413a1} over all
probability measures $P_X$ such that $0 < K_\alpha(P_X \| Q_X) < \infty$
gives the upper bound on $\mu_{k_\alpha}(Q_X, W_{Y|X})$ in \eqref{eq: 20200412d1}.
The proof of the upper bound on $\mu_{s_\alpha}(Q_X, W_{Y|X})$, for all $\alpha \in [0,1]$,
follows similarly from \eqref{eq: 20200412c2}, since the function $g_\alpha(\cdot)$
as defined in \eqref{eq: 20200412c3} is positive over the interval $(0,1)$.

\subsection{Proof of Corollary~\ref{corollary: bounds - contraction coefficient}}
\label{proof of Corollary: bounds - contraction coefficient}

The upper bounds in \eqref{eq: 20200417a1} and \eqref{eq: 20200417a2} rely on those in \eqref{eq: 20200412d1}
and \eqref{eq: 20200412d2}, respectively, by showing that
\begin{align} \label{eq: 20200417a3}
\sup_{s \in (0,1]} \mu_{\phi_s}(Q_X, W_{Y|X}) \leq \mu_{\chi^2}(W_{Y|X}).
\end{align}
Inequality \eqref{eq: 20200417a3} is obtained as follows, similarly to the concept of the
proof of \cite[Remark~3.8]{Raginsky16}. For all $s \in (0,1]$ and $P_X \neq Q_X$,
\begin{align}
& \frac{D_{\phi_s}(P_X W_{Y|X} \, \| \, Q_X W_{Y|X})}{D_{\phi_s}(P_X \| Q_X)} \nonumber \\
\label{eq: 20200417a4}
&= \frac{\chi^2( P_X W_{Y|X} \, \| \, (1-s) P_X W_{Y|X} + s Q_X W_{Y|X})}{\chi^2( P_X \, \| \, (1-s) P_X + s Q_X)} \\
&\leq \mu_{\chi^2}((1-s) P_X + s Q_X, \, W_{Y|X}) \\
\label{eq: 20200417a5}
&\leq \mu_{\chi^2}(W_{Y|X}),
\end{align}
where \eqref{eq: 20200417a4} holds due to \eqref{eq: GV01-chi^2},
and \eqref{eq: 20200417a5} is due to the definition in \eqref{eq: 20200417a0}.

\subsection{Proof of Proposition~\ref{proposition: bounds - contraction coefficient}}
\label{proof of proposition: bounds - contraction coefficient}

The lower bound on the contraction coefficients in \eqref{eq: 20200417b1}
and \eqref{eq: 20200417b2} is due to \eqref{eq: CohenKZ98}.
The derivation of the upper bounds relies on \cite[Theorem~2.2]{Makur_PhD19}, which states the following.
Let $f \colon [0, \infty) \to \Reals$ be a three-times differentiable, convex function with $f(1)=0$, $f''(1) > 0$,
and let the function $z \colon (0, \infty) \to \Reals$ defined as $z(t) := \frac{f(t)-f(0)}{t}$, for all $t>0$,
be concave. Then,
\begin{align} \label{eq: Makur's bound}
\mu_f(Q_X, W_{Y|X}) \leq \frac{f'(1)+f(0)}{f''(1) \, Q_{\min}} \cdot \mu_{\chi^2}(Q_X, W_{Y|X}).
\end{align}
For $\alpha \in (0,1]$, let $z_{\alpha,1} \colon (0, \infty) \to \Reals$ and
$z_{\alpha,2} \colon (0, \infty) \to \Reals$ be given by
\begin{align}
z_{\alpha,1}(t) &:= \frac{k_\alpha(t) - k_\alpha(0)}{t}, \quad t>0,\\
z_{\alpha,2}(t) &:= \frac{s_\alpha(t) - s_\alpha(0)}{t}, \quad t>0,
\end{align}
with $k_\alpha$ and $s_\alpha$ in \eqref{eq: 20200412b1} and \eqref{eq: 20200412b2}.
Straightforward calculus shows that, for $\alpha \in (0,1]$ and $t>0$,
\begin{align}
\tfrac1{\log \mathrm{e}} \, z_{\alpha,1}''(t)
&= -\frac{\alpha^2 + 2 \alpha (1-\alpha)t}{t^2 \bigl[\alpha + (1-\alpha)t\bigr]^2} < 0, \\
\label{eq: 20200417b3}
\tfrac1{\log \mathrm{e}} \, z_{\alpha,2}''(t)
&= -\frac{\alpha^2 \bigl[\alpha + 2(1-\alpha)t \bigr]}{t^2 \bigl[\alpha + (1-\alpha)t \bigr]^2} \\
& \hspace*{0.4cm} - \frac{2(1-\alpha)}{t^3} \biggl[ \log_{\mathrm{e}} \biggl(1 + \frac{(1-\alpha)t}{\alpha} \biggr)
- \frac{(1-\alpha)t}{\alpha+(1-\alpha)t} - \frac{(1-\alpha)^2 t^2}{2 \bigl[\alpha+(1-\alpha)t\bigr]^2} \Biggr]. \nonumber
\end{align}
The first term in the right side of \eqref{eq: 20200417b3} is negative. For showing that the second term
is also negative, we rely on the power series expansion
$\log_{\mathrm{e}}(1+u) = u - \tfrac12 u^2 + \tfrac13 u^3 - \ldots$ for $u \in (-1, 1]$.
Setting $u:=-\frac{x}{1+x}$, for $x>0$, and using Leibnitz theorem for alternating series yields
\begin{align} \label{eq: 20200418b1}
\log_{\mathrm{e}}(1+x) = -\log_{\mathrm{e}}\biggl(1-\frac{x}{1+x} \biggr)
> \frac{x}{1+x} + \frac{x^2}{2(1+x)^2}, \qquad x>0.
\end{align}
Consequently, setting $x:= \frac{(1-\alpha)t}{\alpha} \in [0,\infty)$ in \eqref{eq: 20200418b1},
for $t>0$ and $\alpha \in (0,1]$, proves that the second term in the right side of \eqref{eq: 20200417b3}
is negative. Hence, $z_{\alpha,1}''(t), \, z_{\alpha,2}''(t) < 0$,
so both $z_{\alpha,1}, z_{\alpha,2} \colon (0, \infty) \to \Reals$ are concave functions.

In view of the satisfiability of the conditions of \cite[Theorem~2.2]{Makur_PhD19} for the $f$-divergences
with $f = k_\alpha$ or $f = s_\alpha$, the upper bounds in \eqref{eq: 20200417b1} and \eqref{eq: 20200417b2}
follow from \eqref{eq: Makur's bound}, and also since
\begin{align}
& k_\alpha(0) = 0,  \hspace*{2.7cm} k'_\alpha(1) = \alpha \, \log \mathrm{e},
\hspace*{1.5cm} k''_\alpha(1) = \alpha^2 \, \log \mathrm{e}, \\
& s_\alpha(0) = -(1-\alpha) \, \log \alpha,  \quad
s'_\alpha(1) = (2\alpha-1) \, \log \mathrm{e}, \quad
s''_\alpha(1) = (1-3\alpha+3\alpha^2) \, \log \mathrm{e}.
\end{align}

\subsection{Proof of Proposition~\ref{prop: max. correlation}}
\label{proof of prop: max. correlation}

In view of \eqref{identity2: KL}, we get
\begin{align}
\label{20200413a1}
\frac{D(P_Y \| Q_Y)}{D(P_X \| Q_X)}
&= \frac{\int_0^1 \chi^2(P_Y \, \| \, (1-s)P_Y + sQ_Y)
\; \frac{\mathrm{d}s}{s}}{\int_0^1
\chi^2(P_X \, \| \, (1-s)P_X + sQ_X) \; \frac{\mathrm{d}s}{s}} \\
\label{20200413a2}
&\leq \frac{\int_0^1 \mu_{\chi^2}((1-s)P_X + sQ_X, \, W_{Y|X}) \;
\chi^2(P_X \, \| \, (1-s)P_X + sQ_X)
\; \frac{\mathrm{d}s}{s}}{\int_0^1
\chi^2(P_X \, \| \, (1-s)P_X + sQ_X) \; \frac{\mathrm{d}s}{s}} \\
\label{20200413a3}
&\leq \sup_{s \in [0,1]} \mu_{\chi^2}((1-s)P_X + sQ_X, \, W_{Y|X}).
\end{align}
In view of \eqref{eq:Sarmanov}, the distributions of the random
variables $X_s$ and $Y_s$, and since the equality
$\bigl((1-s)P_X + sQ_X\bigr) W_{Y|X} = (1-s)P_Y + sQ_Y$ holds
for all $s \in [0,1]$, it follows that
\begin{align}
\label{20200413a4}
\rho_{\mathrm{m}}(X_s;Y_s) = \sqrt{\mu_{\chi^2}((1-s)P_X + sQ_X, \, W_{Y|X})}, \quad s \in [0,1],
\end{align}
which, from \eqref{20200413a1}--\eqref{20200413a4}, implies that
\begin{align}
\label{20200413a5}
\sup_{s \in [0,1]} \rho_{\mathrm{m}}(X_s;Y_s) \geq \sqrt{\frac{D(P_Y \| Q_Y)}{D(P_X \| Q_X)}}.
\end{align}
Switching $P_X$ and $Q_X$ in \eqref{20200413a1}--\eqref{20200413a3} and using the
mapping $s \mapsto 1-s$ in \eqref{20200413a3} gives (due to the symmetry of the
maximal correlation)
\begin{align}
\label{20200413a6}
\sup_{s \in [0,1]} \rho_{\mathrm{m}}(X_s;Y_s) \geq \sqrt{\frac{D(Q_Y \| P_Y)}{D(Q_X \| P_X)}},
\end{align}
and, finally, taking the maximal lower bound among those in \eqref{20200413a5}
and \eqref{20200413a6} gives \eqref{eq: LB max. correlation}.

\subsection*{Acknowledgment}
Sergio Verd\'{u} is gratefully acknowledged for a careful reading, and
well-appreciated feedback on the submitted version of this paper.

\end{document}